\documentclass[preprint]{imsart}
\usepackage{amsthm,amsmath}
\usepackage[numbers]{natbib}
\RequirePackage[colorlinks,citecolor=blue,urlcolor=blue,plainpages=false,pdfpagelabels]{hyperref}
\usepackage{geometry}                
\usepackage{graphicx}
\usepackage{amssymb}
\usepackage{epstopdf}
\usepackage{xargs}
\usepackage{ifthen}
\usepackage{dsfont}
\usepackage{aliascnt}
\usepackage{stmaryrd}
\usepackage{mathtools}
\usepackage{todonotes}
\usepackage{enumerate}
\usepackage{ushort}
\usepackage{algpseudocode}
\usepackage{algorithm}
\usepackage{tikz}
\usetikzlibrary{positioning}
\usepackage{autonum}
\numberwithin{equation}{section}

\startlocaldefs

\newtheorem{theorem}{Theorem}
\newaliascnt{proposition}{theorem}
\newtheorem{proposition}[proposition]{Proposition}
\aliascntresetthe{proposition}

\newaliascnt{lemma}{theorem}
\newtheorem{lemma}[lemma]{Lemma}
\aliascntresetthe{lemma}

\newaliascnt{corollary}{theorem}
\newtheorem{corollary}[corollary]{Corollary}
\aliascntresetthe{corollary}

\newaliascnt{definition}{theorem}

\aliascntresetthe{definition}

\newaliascnt{example}{theorem}

\aliascntresetthe{example}

\newaliascnt{remark}{theorem}
\newtheorem{remark}[remark]{Remark}
\aliascntresetthe{remark}

\newtheorem{assumption}{Assumption}

\newcounter{hypA}

\newcounter{hypB}

\def\equationautorefname~#1\null{%
  Equation~(#1)\null
}

\newcommand{\1}[1]{\mathds{1}_{#1}}

\newcommandx\A[2][1=]{
\ifthenelse{\equal{#1}{}}
{\hspace{-1mm}(\textbf{A\ref{#2}})\hspace{-1mm}}
{\hspace{-1mm}(\textbf{A\ref{#1}--\ref{#2}})\hspace{-1mm}}
}
\newcommand{\af}[1]{h_{#1}} 
\newcommand{\alg}[1]{\mathcal{#1}} 

\newcommand{\addf}[1]{\termletter_{#1}} 

\newcommandx{\arr}[2][1=]{
\ifthenelse{\equal{#1}{}}
{\upsilon_{\N}^{#2}}
{(\upsilon_{\N}^{#2})^{#1}}
}
\newcommandx{\arrterm}[3][1=]{
\ifthenelse{\equal{#1}{}}
{\tilde{\upsilon}_{\N}(#2,#3)}
{\tilde{\upsilon}_{\N}^{#1}(#2,#3)}
}
\newcommandx{\asvar}[4][1=]{
\ifthenelse{\equal{#1}{}}
{\sigma_{#2} \langle #3, #4 \rangle(\af{#2})}
{\sigma_{#2}^2 \langle #3, #4 \rangle(\af{#2})}
}
\newcommandx{\asvarFFBSm}[4][1=]{
\ifthenelse{\equal{#1}{}}
{\tilde{\sigma}_{#2} \langle #3, #4 \rangle(\af{#1})}
{\tilde{\sigma}_{#2}^2 \langle #3, #4 \rangle(\af{#2})}
}
\newcommandx{\asvarstd}[2][1=]{
\ifthenelse{\equal{#1}{}}
{\sigma_{#2}(\af{#2})}
{\sigma_{#2}^2(\af{#2})}
}
\newcommandx{\asvarFFBSmstd}[2][1=]{
\ifthenelse{\equal{#1}{}}
{\tilde{\sigma}_{#2}(\af{#2})}
{\tilde{\sigma}_{#2}^2(\af{#2})}
}
\newcommandx{\asvard}[3][1=]{
\ifthenelse{\equal{#1}{}}
{\sigma_{#2}(#3)}
{\sigma_{#2}^2 (#3)}
}
\newcommandx{\asvardFFBSm}[3][1=]{
\ifthenelse{\equal{#1}{}}
{\sigma_{#2} \langle #3 \rangle(\af{#2})}
{\sigma_{#2}^2 \langle #3 \rangle(\af{#2})}
}
\newcommand{\asvardPaRIS}[4]{\varsigma_{#1,#2,#3;\parvec}(#4)}

\newcommandx\B[2][1=]{
\ifthenelse{\equal{#1}{}}
{\hspace{-1mm}(\textbf{B\ref{#2}})\hspace{-1mm}}
{\hspace{-1mm}(\textbf{B\ref{#1}--\ref{#2}})\hspace{-1mm}}
}
\newcommandx{\BF}[3][1=]{
\ifthenelse{\equal{#1}{}}
{\kernel{D}_{#2, #3}}
{\kernel{D}_{#2, #3}^{#1}}
}
\newcommandx{\BFcent}[3][1=]{
\ifthenelse{\equal{#1}{}}
{\tilde{\kernel{D}}_{#2, #3}}
{\tilde{\kernel{D}}_{#2, #3}^{#1}}
}
\newcommand{\bi}[3]{J_{#1}^{(#2, #3)}}
\newcommandx{\bk}[2][1=]{ 
\ifthenelse{\equal{#1}{}}
{\overleftarrow{\kernel{Q}}_{#2}}
{\overleftarrow{\kernel{Q}}_{#2}^{#1}}
}

\newcommand{\bmf}[1]{\set{F}(#1)} 
\newcommand{\borel}[1]{\mathcal{B}(\set{#1})} 

\newcommandx{\cexp}[3][1=]{
\ifthenelse{\equal{#1}{}}
{\mathbb{E}\left[ #2 \mid #3 \right]} 
{\mathbb{E}[ #2 \mid #3 ]}
}

\newcommand{\convd}{\overset{\mathcal{D}}{\longrightarrow}}
\newcommand{\convp}{\overset{\prob}{\longrightarrow}}

\newcommand{\defeq}{=\vcentcolon}
\newcommand{\deriv}{\nabla_{\hspace{-.07cm}\parvec}}

\newcommand{\epart}[2]{\xi_{#1}^{#2}}
\newcommand{\eqdef}{\vcentcolon=} 
\newcommand{\eqsp}{}

\newcommandx\filtderiv[2][1=]{
\ifthenelse{\equal{#1}{}}
	{\eta_{#2}}
	{\eta_{#2}^\N}
}

\newcommandx{\genfd}[1][1=]{
\ifthenelse{\equal{#1}{}}
{\mathcal{F}}
{\mathcal{F}_{\N}}
}

\newcommand{\hbd}{| \tilde{h} |_{\infty}}
\newcommand{\hk}{\kernel{Q}} 

\newcommand{\hkup}{\bar{\varepsilon}}
\newcommand{\hd}[1][]{q_{#1}} 

\newcommand{\idop}{\operatorname{id}}

\newcommand{\jimmy}[1]{#1}
\newcommand{\johan}[1]{#1}
\renewcommand{\k}{j}
\newcommand{\kernel}[1]{\mathbf{#1}}
\newcommand{\kletter}{\tilde{\N}}
\newcommandx{\K}[1][1=]{
\ifthenelse{\equal{#1}{}}{{\kletter}}{{\tilde{\N}^{#1}}}
}

\newcommandx{\lebfun}[1][1=]{
\ifthenelse{\equal{#1}{}}
{\lebfunletter}
{\lebfunletter_{#1}}
}

\newcommand{\lebfunletter}{\varphi}
\newcommand{\lk}[1]{\kernel{L}_{#1}} 

\newcommand{\llh}[1]{L_{#1}}
\newcommand{\logllh}[1]{\ell_{#1}}

\newcommandx{\M}[1][1=]{
\ifthenelse{\equal{#1}{}}
{\N}
{\N}
}
\newcommand{\md}[1]{g_{#1}} 
\newcommand{\mdlow}{\ushort{\delta}}
\newcommand{\mdup}{\bar{\delta}}
\newcommand{\meas}[1]{\mathsf{M}(#1)}
\newcommand{\measSpace}[1]{(\set{#1},\alg{#1})} 
\newcommand{\mr}{\varrho}

\newcommand{\mk}{\kernel{G}}

\newcommand{\N}{N}

\newcommand{\nset}{\mathbb{N}}
\newcommand{\nsetpos}{\mathbb{N}^*}

\newcommandx{\oscn}[2][1=]{
\ifthenelse{\equal{#1}{}}{\operatorname{osc}(#2)}{\operatorname{osc}^{#1}(#2)}
}

\newcommand{\partpred}[1]{\pi_{#1}^\N}

\newcommandx\post[2][1=]{
\ifthenelse{\equal{#1}{}}
	{\phi_{#2}}
	{\phi_{#2}^\N}
}

\newcommandx\postafl[2][1=]{
\ifthenelse{\equal{#1}{}}
	{\phi_{#2}^{\tol}}
	{\phi_{#2}^{\N,\tol}}
}

\newcommand{\parvec}{\theta}
\newcommand{\parspace}{\Theta}
\newcommand{\pred}[1]{\pi_{#1}}
\newcommand{\prob}{\mathbb{P}} 
\newcommand{\probdist}{\mathsf{Pr}}
\newcommand{\probmeas}[1]{\mathsf{M}_1(#1)}

\newcommand{\refM}{\mu} 
\newcommand{\rmd}{\mathrm{d}}
\newcommand{\rmlterm}[2]{\hat{\zeta}_{#1}^{#2}}
\newcommand{\rset}{\mathbb{R}}

\newcommand{\set}[1]{\mathsf{#1}} 

\newcommand{\stat}{\Pi}
\newcommand{\supn}[1]{\|#1\|_{\infty}} 

\newcommand{\termletter}{\tilde{h}}
\newcommand{\testfsymb}{f}
\newcommandx{\testf}[1][1=]{  
\ifthenelse{\equal{#1}{}}{\testfsymb}{\testfsymb_{#1}}
}

\newcommand{\testfpsymb}{\mathfrak{F}}
\newcommandx{\testfp}[1][1=]{  
	\ifthenelse{\equal{#1}{}}{\testfpsymb}{\testfpsymb_{#1}}
}
\newcommandx{\testfc}[1][1=]{
	\ifthenelse{\equal{#1}{}}{\underline{\testfsymb}}{\underline{\testfsymb}_{#1}}
}
\newcommand{\testfhsymb}{\tilde{\mathfrak{F}}}
\newcommandx{\testfh}[1][1=]{  
	\ifthenelse{\equal{#1}{}}{\testfhsymb}{\testfhsymb_{#1}}
}

\newcommand{\tol}{\varepsilon}

\newcommand{\truepar}{{\parvec_\ast}}
\newcommand{\tstatletter}{\kernel{T}}
\newcommandx\tstat[2][1=]{
\ifthenelse{\equal{#1}{}}
	{\tstatletter_{#2}}
	{\tau_{#2}^{#1}}
}
\newcommand{\tstattil}[2]{\tilde{\tau}_{#2}^{#1}}

\newcommandx\varlim[2][1=]{
\ifthenelse{\equal{#1}{}}
	{\sigma_{#2}^{2, \infty}}
	{\sigma_{#2}^{2, \N}}
}

\newcommand{\wgt}[2]{\omega_{#1}^{#2}}
\newcommand{\wgtsum}[1]{\Omega_{#1}}

\newcommand{\Xinit}{\chi}

\endlocaldefs

\begin{document}

\begin{frontmatter}


\title{Particle-based, online estimation of tangent filters with application to parameter estimation in nonlinear state-space models}
\runtitle{Online estimation of tangent filters}

\begin{aug}
\author{\fnms{Jimmy} \snm{Olsson}
\ead[label=e1]{jimmyol@kth.se}}
\and
\author{\fnms{Johan} \snm{Westerborn Alenl\"ov}
\ead[label=e2]{johawes@kth.se}}

\runauthor{J.~Olsson and J.~Westerborn~Alenl\"ov}

\affiliation{KTH Royal Institute of Technology}

\address{Department of Mathematics \\
KTH Royal Institute of Technology \\
SE-100 44  Stockholm, Sweden \\
\printead{e1}, \printead*{e2}}
\end{aug}


\begin{abstract}
This paper presents a novel algorithm for efficient online estimation of the filter derivatives in general hidden Markov models. The algorithm, which has a linear computational complexity and very limited memory requirements, is furnished with a number of convergence results, including a central limit theorem with an asymptotic variance that can be shown to be uniformly bounded in time. Using the proposed filter derivative estimator we design a recursive maximum likelihood algorithm updating the parameters according the gradient of the one-step predictor log-likelihood. The efficiency of this online parameter estimation scheme is illustrated in a simulation study. 
\end{abstract}

\end{frontmatter}


\section{Introduction}
\label{sec:introduction}

The general state-space \emph{hidden Markov models} form a powerful statistical modeling tool which is presently applied across a wide range of scientific and engineering disciplines; see e.g. \cite{cappe:2001:hmmbib,cappe:moulines:ryden:2005} and the references therein for such examples. 
In the literature, the term ``hidden Markov model'' is often restricted to the case where the state space of the hidden chain is a finite set, and we hence use  the term ``state-space model'' (SSM) to stress that the state spaces of the models that we consider are completely general. More specifically, SSMs are bivariate stochastic processes consisting of an observable process $\{ Y_t \}_{t \in \nset}$ and an unobservable Markov chain $\{ X_t \}_{t \in \nset}$, referred to as the \emph{state process} and \emph{observation process} and taking on values in some general state spaces $\set{Y}$ and $\set{X}$, respectively. When using these models in practice, the statistician is typically interested in calculating conditional distributions of the unobservable states given some fixed observation record or data-stream $\{ y_t \}_{t \in \nset}$.  Generally, these distributions cannot be expressed in a closed form, and the user needs to rely on approximations. Moreover, any SSM is typically parameterised by unknown model parameters $\parvec$ which need to be calibrated ex ante.  

The frequentist approach to parameter estimation consists in maximising the \emph{likelihood function} $\parvec \mapsto \llh{\parvec}(y_{0:t})$, where $ \llh{\parvec}(y_{0:t})$ denotes the joint density of the observations evaluated at the given data $y_{0:t} = (y_0, \ldots, y_t)$ (this will be our generic notation for vectors), or, equivalently, the \emph{log-likelihood} function $\parvec \mapsto \logllh{\parvec}(y_{0:t}) \eqdef \log \llh{\parvec}(y_{0:t})$. Even though the likelihood cannot generally be expressed in a closed form, there are several approaches to approximate maximisation of the same. These methods are typically based on either the \emph{expectation-maximisation} algorithm~\citep{dempster:laird:rubin:1977}, where the maximisation is performed over an intermediate quantity, or the tools of \emph{gradient-based optimisation}. The latter algorithms produce, in a \emph{steepest ascent} manner, a sequence $\{Ê\parvec_n \}_{n \in \nset}$ of parameter estimates converging to the maximum likelihood estimator; i.e., at each iteration, a parameter estimate $\parvec_n$ is updated recursively through a move in the direction given by an estimate of the \emph{score function} $\deriv \logllh{\parvec_n}(y_{0:t})$. 

There are mainly two approaches to estimation of the score function. The first one goes via \emph{Fisher's identity} \citep[see, e.g.,][Proposition~10.1.6]{cappe:moulines:ryden:2005}, which relates the score function of the observed data to that of the complete data and thus transfers the problem to that of \emph{joint smoothing}, i.e., the computation of the joint posterior distribution of $X_{0:t}$ given $Y_{0:t} = y_{0:t}$. The second approach, which is the focus of the present paper, goes via the decomposition 
\begin{equation} \label{eq:log-likelihood:sum}
	\logllh{\parvec}(y_{0:t}) = \sum_{s = 0}^t \logllh{\parvec}(y_s \mid y_{0:s - 1}),
\end{equation}
where each \emph{one-step predictor likelihood} $\logllh{\parvec}(y_s \mid y_{0:s - 1})$ (with $\logllh{\parvec}(y_{0} \mid y_{0:-1}) \eqdef \logllh{\parvec}(y_0)$) can be computed via the \emph{sensitivity equations}; see~\citet[Section~10.2.4]{cappe:moulines:ryden:2005} and~\autoref{sec:par:est} for details. Let $\pred{s;\parvec}(x_s)$ be the density of the \emph{predictor} at time $s$, i.e., the distribution of the state $X_s$ conditional to the observation record $Y_{0:s - 1} = y_{0:s - 1}$; it is then, by swapping the order of the nabla and integration operations, possible to express the gradient $\deriv \logllh{\parvec}(y_s \mid y_{0:s - 1})$ as a functional of the gradient  $\deriv \pred{s;\parvec}(x_s)$, the latter being typically referred to as the \emph{filter derivative} or \emph{tangent filter} or \emph{filter sensitivity}. It should be noted that despite its name, the tangent filter is the gradient of the \emph{predictor} rather than the \emph{filter}, where the filter at time $s$ is defined as the distribution of $X_s$ conditionally to $Y_{0:s} = y_{0:s}$, i.e., the updated predictor. 

Appealingly, the recursive updating scheme for the tangent filters given by the sensitivity equations requires only access to the filter distributions (and not the full joint smoothing distribution, as in the case of Fisher identity-based score approximation). If estimation is to be carried through in the \emph{offline} mode, by processing repeatedly a fixed batch of observations until convergence, the approach based on Fisher's identity is typically computationally more efficient and hence to be preferred. On the other hand, if estimation is to be carried through in an \emph{online} manner, i.e., through a single sweep of a (potentially infinitely) long sequence of observations, then the sensitivity equation approach is better suited, as the filter distributions can be updated recursively at constant computational cost. Online parameter learning is relevant when dealing with very big batches of data or in scenarios requiring the parameters to be continuously updated on the basis of a continuous stream of arriving observations.

If the state space of the latent Markov chain is a finite set or if the model belongs to the linear Gaussian SSMs, exact computation of the filter distributions---and thus the filter derivatives and the gradient of the one-step predictor log-likelihood---is possible \citep{legland:mevel:1997,cappe:moulines:ryden:2005}. However, as touched upon above, the general case calls for approximation of these quantities, and existing approaches to nonlinear filtering can typically be divided into two classes. The first class contains methods relying on linearisation, such as the \emph{extended Kalman filter}, the \emph{unscented Kalman filter}~\citep{anderson:moore:1979,julier:uhlmann:1997}, etc., whose success is generally limited in the presence of significant nonlinear or non-Gaussian model components. The second class is formed by the \emph{sequential Monte Carlo} (SMC) or \emph{particle filtering} approaches \citep{gordon:salmond:smith:1993}, where a sample of particles with associated importance weights is propagated sequentially and randomly in time according to the model dynamics in order to approximate the flows of filter and predictor distributions. Due to their high capability of solving potentially very difficult nonlinear/non-Gaussian filtering problems, SMC methods have, as indicated by the over 8000 Google Scholar citations of the first monograph \cite{doucet:defreitas:gordon:2001} on the topic, seen a rapid and ever-increasing interest during the last decades. 

\subsection{Previous work}

For an expos\'e of current particle-based approaches to parameter estimation in nonlinear SSMs, we refer to \cite{kantas:doucet:singh:maciejowski:chopin:2015}. SMC-based score function estimation via tangent filters was discussed by~\cite{doucet:tadic:2003,poyiadjis:doucet:singh:2005,poyiadjis:doucet:singh:2011,delmoral:doucet:singh:2015}. 

The work \cite{delmoral:doucet:singh:2015}, which can be viewed as the state of the art when it concerns SMC-based tangent filter estimation, expresses each tangent filter as a centered, smoothed expectation of an additive state functional. Using the so-called \emph{backward decomposition} of the joint smoothing distribution and a recursive form of this decomposition going back to \cite{cappe:2009}, {\cite{delmoral:doucet:singh:2010} are able to form, on-the-fly and through the perspicacious propagation of a number of smoothed statistics, one statistic per particle, non-collapsed particle approximations of such additive smoothed expectations. Moreover, \cite{delmoral:doucet:singh:2010}} furnish their estimator with a central limit theorem (CLT) describing the weak convergence, as the size of the particle sample tends to infinity, of the particle tangent filter approximations to the exact counterparts. Appealingly, the fact that the smoothed expectations are centered allows, under strong mixing assumptions, the asymptotic variance of the mentioned CLT to be bounded uniformly in time, which establishes the long-term stochastic stability of the algorithm for large sample sizes. {However, since the propagation of the smoothed statistic associated with a given particle requires an expectation under the ancestral posterior of that particle to be computed by summation, the computational complexity of the algorithm is \emph{quadratic} in the number of particles. The} user is hence forced to keep the number of particles at a modest value, which may imply {severe} numerical instability in practice (see \autoref{sec:simulations} for an illustration). 

\subsection{Our contribution}

On the basis of the approach taken by \cite{delmoral:doucet:singh:2015}, we develop an efficient and numerically stable algorithm for particle approximation of the filter derivatives in SSMs. The approach that we propose includes, as a sub-routine, the \emph{particle-based, rapid incremental smoother} (PaRIS) \citep[introduced recently by us in][]{olsson:westerborn:2014b}, an algorithm that is tailor-made for online approximation of smoothed additive state functionals.

{By replacing, in the forward recursion proposed by \cite{delmoral:doucet:singh:2015}, the computation of each smoothed statistic by a conditionally unbiased Monte Carlo estimate, the computational complexity \johan{becomes}, under mild assumptions, \johan{\emph{linear}} in the number of particles. Since the original estimator can be viewed as a Rao-Blackwellisation of the new one, we will sometimes refer to this measure as ``\emph{de}-Rao-Blackwellisation''.} This allows the user to mount a considerably larger amount of particles for a given computational budget, which compensates by far for the slight increase of asymptotic variance added by the additional sampling step. {Thus, the technique that we propose has some similarities with that used in particle-based \emph{two-filter smoothing}, where a linear computational complexity is obtained by replacing certain marginalisations by random sampling \citep[see][]{fearnhead:wyncoll:tawn:2010} (however, for the sake of completeness we remark that two-filter particle smoothers, which require a so-called \emph{backward information filter} to be run backwards in time, cannot be used in online settings).} By reusing techniques developed by \cite{olsson:westerborn:2014b} and \citet{delmoral:doucet:singh:2015} we are able to furnish the proposed algorithm with an exponential concentration inequality and a CLT, whose asymptotic variance is given by that of the Rao-Blackwellised counterpart plus one additional term being inversely proportional to sample size of the supplementary, PaRISian sampling operation. {This result is analogous to the results obtained by \cite{nguyen:lecorff:moulines:2017} for the case of particle-based two-filter smoothing.} Moreover, by assuming strong mixing of the model, the asymptotic variance can be bounded uniformly in time. 

In the second part of the paper, we cast our PaRIS-based tangent filter estimator into the framework of recursive maximum likelihood estimation. {The numerical performance of this online parameter estimation algorithm is investigated in a simulation study, where we are able to report significant improvement over the approach of \cite{delmoral:doucet:singh:2015} for online parameter learning in the stochastic volatility model of \cite{hull:white:1987}. 
In addition, we propose our algorithm as a prototypical solution to the \emph{simultaneous localisation and mapping} (SLAM) problem, with promising results.
} 

Finally, we remark that the parameter estimation algorithm was outlined by us in the conference note~\cite{olsson:westerborn:2016}. The present paper provides rigorous theoretical results underpinning the initial observations made by \cite{olsson:westerborn:2016} and complement the simulation study of that work with additional examples. 

\subsection{Outline}

After having introduced, prefatorily, general SSMs, tangent filters, and SMC-based smoothing (including the PaRIS) in~\autoref{sec:preliminaries}, \autoref{sec:main:results} presents the proposed approach to tangent filter estimation. In addition, \autoref{sec:algo} includes theoretical convergence results in the form of an exponential concentration inequality and a CLT with a time-uniform bound on the asymptotic variance. \autoref{sec:par:est} applies the PaRISian tangent filter estimator to recursive maximum likelihood estimation, and the resulting online parameter estimation algorithm is tested numerically in \autoref{sec:simulations}. Finally, some conclusions are drawn in~\autoref{sec:discussion} and \autoref{sec:proofs} contains all proofs.


\section{Preliminaries}
\label{sec:preliminaries}

\subsection{General state-space models}
\label{sec:notation}

In the following we let $\nset$ denote the natural numbers and set $\nsetpos \eqdef \nset \setminus \{0\}$. For any measurable space $(\mathsf{E}, \mathcal{E})$, where $\mathcal{E}$ is a countably generated $\sigma$-algebra, we denote by $\bmf{\mathcal{E}}$ the set of bounded $\mathcal{E}/\borel{\mathbb{R}}$-measurable functions on $\mathsf{E}$, where $\borel{\mathbb{R}}$ denotes the Borel $\sigma$-algebra. For any function $h \in \bmf{\mathcal{E}}$, we let $\supn{h} \eqdef \sup_{x \in \mathsf{E}}|h(x)|$ denote the supremum norm of $h$. Let $\meas{\mathcal{E}}$ denote the set of $\sigma$-finite measures on $\mathcal{E}$ and $\probmeas{\mathcal{E}} \subset \meas{\mathcal{E}}$ the set of probability measures on the same space. 

Now, let $(\set{X}, \alg{X})$ and $(\set{Y}, \alg{Y})$ be measurable spaces and $\hk_{\parvec} : \set{X} \times \alg{X} \to [0,1]$ and $\mk_{\parvec} : \set{X} \times \alg{Y} \to [0,1]$ Markov transition kernels. We will often subject kernels to operations such as multiplication and tensor multiplication. {Moreover, a transition kernel induces two different operators, one on bounded measurable functions and one on measures; we refer to \autoref{sec:kernel} for details.} Moreover, let $\Xinit \in \probmeas{\alg{X}}$. The kernels above are parameterised by a model parameter vector $\parvec$ belonging to some space {$\parspace \subseteq \rset^d$, \johan{$d \in \nsetpos$},} and we will use subindices to stress the dependence of any distributions on these parameters. We define an SSM as the canonical Markov chain $\{(X_t, Y_t)\}_{t \in \nset}$ with initial distribution $\Xinit \varotimes \mk_{\parvec}$ \johan{(see~\autoref{sec:kernel} for a definition of the product $\varotimes$ of a measure and a kernel)} and Markov transition kernel 
\begin{equation} \label{eq:kernel:SSM}
    (\set{X} \times \set{Y}) \times (\alg{X} \varotimes \alg{Y}) \ni ((x, y), A) \mapsto \int \1{A}(x', y') \, \hk_\parvec(x, \rmd x') \, \mk_\parvec(x', \rmd y'). 
\end{equation}
In this setting, the state process $\{X_t\}_{t \in \nset}$ is assumed to be only partially observed through the observation process $\{Y_t\}_{t \in \nset}$. Under the dynamics \eqref{eq:kernel:SSM}, 
\begin{itemize}
	\item[(i)] the state process $\{X_t\}_{t \in \nset}$ is itself a Markov chain with transition kernel $\hk_{\parvec}$ and initial distribution $\Xinit$.  
	\item[(ii)] the observations are, conditionally on the states, independent and such that the conditional distribution of each $Y_t$ depends on the corresponding $X_t$ only and is given by the \emph{emission distribution} $\mk_{\parvec}(X_t, \cdot)$.
\end{itemize}
For simplicity we have assumed that the initial distribution $\Xinit$ does not depend on $\parvec$. Throughout this paper we will consider the {fully dominated case where for all $\parvec \in \Theta$, $\hk_{\parvec}$ and $\mk_{\parvec}$ admit densities $\hd[\parvec]$ and $\md{\parvec}$ with respect to some reference measures $\refM \in \meas{\alg{X}}$ and $\nu \in \meas{\alg{Y}}$, respectively. This means that for all $\theta \in \Theta$, $x \in \set{X}$, $f \in \bmf{\alg{X}}$, and $h \in \bmf{\alg{Y}}$,
\begin{align}
	\hk_{\parvec}f(x) &= \int f(x') \hd[\parvec](x, x') \, \refM(\rmd x'), \\
	\mk_{\parvec}h(x) &= \int h(y) \md{\parvec}(x, y) \, \nu(\rmd y). 	
\end{align}
Likewise, $\Xinit$ is assumed to have a density, which we denote by the same symbol, with respect to the same reference measure $\refM$.} 

In the following we assume that we are given a {fixed} sequence $\{y_t\}_{t \in \nset}$ of observations of $\{Y_t\}_{t \in \nset}$, and in order to avoid notational overload we will generally suppress any dependence on these observations from the notation by denoting, e.g., by $\md{t ; \parvec}$ the function $\set{X} \ni x \mapsto \md{\parvec}(x, y_t)$. For any triple $(s,s',t) \in \nset^3$ such that $s \leq s'$ we denote by $\post{s:s' \mid t; \parvec}$ the conditional distribution of $X_{s:s'}$ conditioned on $Y_{0:t}$. For any $f \in \bmf{\alg{X}^{\varotimes (s' - s + 1)}}$, this distribution can be expressed as 
{
$$
	\post{s:s' \mid t; \parvec}f = \llh{\parvec}(y_{0:t})^{-1} \idotsint f(x_{s:s'}) \, \left( \prod_{m = 0}^t \md{m ; \parvec}(x_m) \right) \left( \Xinit(\rmd x_0) \prod_{\ell = 0}^{(t \vee s') - 1} \hk_{\parvec}(x_{\ell}, \rmd x_{\ell + 1}) \right),
$$
}
where 
\begin{align}
	\llh{\parvec}(y_{0:t}) = \int \cdots \int \md{0;\parvec}(x_0) \, \Xinit(\rmd x_0) \prod_{\ell = 0}^{t - 1} \md{\ell + 1;\parvec}(x_{\ell + 1}) \, \hk_{\parvec}(x_{\ell}, \rmd x_{\ell + 1}) \label{eq:llh}
\end{align}
is the observed data likelihood. In the cases $s = s' = t$ and $s = s' = t + 1$ we let $\post{t; \parvec}$ and $\pred{t + 1 ; \parvec}$ be shorthand notation for the filter and predictor distributions $\post{t:t \mid t; \parvec}$ and $\post{t + 1:t + 1 \mid t; \parvec}$, respectively. In the case $s = 0$ and $s' = t$, the distribution $\post{0:t \mid t; \parvec}$ is referred to as the joint-smoothing distribution. The filter and predictor distributions are closely related; indeed, the well-known \emph{filtering recursion} provides that for all $t \in \nset$ and $f \in \bmf{\alg{X}}$, 
\begin{align}
	\post{t; \parvec} f &= \frac{\pred{t; \parvec}(\md{t;\parvec} f)}{\pred{t; \parvec} \md{t;\parvec}}, \label{eq:pred:to:filt}\\
	\pred{t + 1; \parvec} f &= \post{t;\parvec} \hk_{\parvec} f, \label{eq:filt:to:pred}
\end{align}
where, by convention, $\pred{0; \parvec} \eqdef \Xinit$. 

We will often deal with sums and products on functions with possibly different arguments. Since these functions will be defined on products of $\set{X}$, we will, when needed, let, with slight abuse of notation, subscripts define the domains of such sums and products. For instance, $\testf[t] \tilde{f}_t : \set{X} \ni x_t \mapsto \testf[t](x_t) \tilde{f}_t(x_t)$, while $\testf[t] + \tilde{f}_{t + 1}: \set{X}^2 \ni x_{t:t+1} \mapsto \testf[t](x_t) + \tilde{f}_{t+1}(x_{t+1})$.

It may be shown \citep[see, e.g.,][Proposition~3.3.6]{cappe:moulines:ryden:2005} that the state process has still the Markov property when evolving conditionally on $Y_{0:t} = y_{0:t}$ in the time-reversed direction. Moreover, the distribution of $X_s$ given $X_{s+1}$ and $Y_{0:t} = y_{0:t}$ is, for $s \leq t$, given by the \emph{backward kernel} denoted by $\bk{\post{s;\parvec}}$. {Under the assumption of a fully dominated model} the backward kernel can, for $x_{s+1} \in \set{X}$ and $\testf \in \bmf{\alg{X}}$, be expressed as
\begin{align}
	\bk{\post{s; \parvec}} f(x_{s+1}) \eqdef \frac{\int f(x_s) \hd[\parvec](x_s, x_{s+1}) \, \post{s; \parvec}(\rmd x_s)}{\int \hd[\parvec](x'_s, x_{s+1}) \, \post{s; \parvec}(\rmd x'_s)}. \label{eq:bk:kernel}
\end{align}
Consequently, using the backward kernel, we may express the joint-smoothing distribution $\post{0:t \mid t; \parvec}$ as
\begin{equation}
	\post{0:t \mid t; \parvec} = \post{t; \parvec} \tstat{t;\parvec}, \label{eq:bk:decomp}
\end{equation}
where we have defined the kernels
\begin{align} 
	\tstat{t ; \parvec} \eqdef 
	\begin{cases} 
	\bk{\post{t-1 ; \parvec}} \varotimes \bk{\post{t-2; \parvec}} \varotimes \cdots \varotimes \bk{\post{0;\parvec}} & \text{for } t \in \nsetpos, \\
	\idop{} & \text{for } t = 0,
	\end{cases}
\end{align}
{where $\varotimes$ denotes kernel tensor product; see \autoref{sec:kernel} for details. Note that for all $x \in \set{X}$, the distribution $\tstat{t ; \parvec}(x, \cdot)$ on the product space $\alg{X}^{\varotimes t}$ describes the law of the inhomogeneous Markov chain initialised at $x \in \set{X}$ and evolving backwards in time according to the backward kernels.}

As we will see in the following, the modeler is often required to compute smoothed expectations of \emph{additive} objective functions of type 
$$
\af{t}(x_{0:t}) \eqdef \sum_{s = 0}^{t - 1} \addf{s}(x_{s:s + 1}), 
$$
{where for all $0 \leq s \leq t - 1$, $\addf{s}$ is a measurable function on $\set{X}^2$.} This setting allows for recursive computation of $\{ \tstat{t; \parvec} \af{t} \}_{t \in \nset}$ according to 
\begin{equation}
	\tstat{t + 1; \parvec}\af{t+1}(x_{t+1}) = \int \{\tstat{t; \parvec}\af{t}(x_t) + \addf{t}(x_{t:t+1})\} \, \bk{\post{t; \parvec}}(x_{t+1}, \rmd x_t); \label{eq:tstat:update}
\end{equation}
see \cite{cappe:2009,delmoral:doucet:singh:2010}. The recursion \eqref{eq:tstat:update} will be a key ingredient also in the coming developments.

\subsection{Tangent filters} 
\label{sec:filter_derivatives}

The gradient of the prediction filter, known as the filter derivative or the tangent filter, is a signed measure which we, following~\cite[Section 2]{delmoral:doucet:singh:2015}, define as follows. We start with considering the gradient of $\post{0:t \mid t - 1; \parvec}$ and then extract the marginal of this quantity. First, 
{note that} for all $f_{0:t} \in \bmf{\alg{X}^{\varotimes (t + 1)}}$, 
\begin{equation} \label{eq:generalised:Fischer}
\frac{1}{\llh{\parvec}(y_{0:t - 1})} \deriv \int f_{0:t}(x_{0:t}) \Xinit(\rmd x_0) \prod_{s = 0}^{t - 1} \md{s;\parvec}(x_{s}) \hd[\parvec](x_{s}, x_{s + 1}) \, \refM^{\varotimes t}(\rmd x_{1:t}) = \post{0:t \mid t - 1; \parvec} (\af{t;\parvec} f_{0:t}), 
\end{equation}
where $\af{t;\parvec}(x_{0:t}) = \sum_{s = 0}^{t - 1} \addf{s;\parvec}(x_{s:s+1})$, with 
\begin{equation} \label{eq:complete:score}
	\addf{s;\parvec}(x_{s:s+1}) = \deriv \log \md{s;\parvec}(x_s) + \deriv \log \hd[\parvec](x_s, x_{s + 1}) {,}
\end{equation}
{is \emph{complete data score function}.} 
{The identity \eqref{eq:generalised:Fischer} is established by, first, swapping, on the left hand side, the order of differentiation and integration and, second, dividing and multiplying the integrand by the complete data likelihood. In the case $f_{0:t} = \1{{\set{X}}^{t + 1}}$, \eqref{eq:generalised:Fischer} coincides with \emph{Fisher's identity}~\citep[see, e.g., ][Proposition 10.1.6]{cappe:moulines:ryden:2005}.}   
Now, using \eqref{eq:generalised:Fischer}, we obtain straightforwardly 
\begin{align}
	\deriv \post{0:t \mid t - 1; \parvec} f_{0:t} &= \deriv \frac{1}{\llh{\parvec}(y_{0:t - 1})} \int f_{0:t}(x_{0:t}) \, \Xinit(\rmd x_0) \prod_{s = 0}^{t - 1} \md{s;\parvec}(x_{s}) \hd[\parvec](x_{s}, x_{s + 1}) \, \refM^{\varotimes t}(\rmd x_{1:t}) \\
	&= \frac{1}{\llh{\parvec}(y_{0:t - 1})} \deriv \int f_{0:t}(x_{0:t}) \, \Xinit(\rmd x_0) \prod_{s = 0}^{t - 1} \md{s;\parvec}(x_{s}) \hd[\parvec](x_{s}, x_{s + 1}) \, \refM^{\varotimes t}(\rmd x_{1:t}) \\
	&- \post{0:t \mid t - 1; \parvec} f_{0:t} \frac{1}{\llh{\parvec}(y_{0:t - 1})}  \deriv \int  \Xinit(\rmd x_0) \prod_{s = 0}^{t - 1} \md{s;\parvec}(x_s) \hd[\parvec](x_s, x_{s + 1}) \, \refM^{\varotimes t}(\rmd x_{1:t}) \\
	&=  \post{0:t \mid t - 1; \parvec} (\af{t;\parvec} f_{0:t}) - \post{0:t \mid t - 1; \parvec} f_{0:t} \times \post{0:t \mid t - 1 ; \parvec} \af{t;\parvec}. \label{eq:tangent:identity:first:step}
\end{align}
In other words, $\deriv \post{0:t \mid t - 1; \parvec} f_{0:t}$ coincides with the covariance of $f_{0:t}$ and the complete data score function $\af{t;\parvec}$ under $\post{0:t \mid t - 1; \parvec}$. {Thus, in order to obtain, from \eqref{eq:tangent:identity:first:step}, an expression of the gradient of the predictor distribution $\pred{t; \parvec}$, which is the marginal of $\post{0:t \mid t - 1; \parvec}$ with respect to the last state, we consider $f_t \in \bmf{\alg{X}}$, let $f_{0:t}(x_{0:t}) = f_t(x_t)$, and use the backward decomposition \eqref{eq:bk:decomp} to obtain}
\begin{equation} \label{eq:tangent:identity}
    \deriv \pred{t; \parvec} f_t = \pred{t; \parvec}\{ (\tstat{t; \parvec} \af{t;\parvec} - \pred{t; \parvec} \tstat{t; \parvec} \af{t;\parvec}) f_t \}.  
\end{equation}
The tangent filter $\filtderiv{t; \parvec}$ is now defined as the signed measure 
$$
    \filtderiv{t; \parvec} f_t \eqdef \deriv \pred{t; \parvec} f_t, \quad f_t \in \bmf{\alg{X}}. 
$$
Since the predictors as well as the backward statistics $\{ \tstat{t;\parvec} \af{t} \}_{t \in \nset}$ may be updated recursively through the filtering recursion \eqref{eq:pred:to:filt}--\eqref{eq:filt:to:pred} and the recursion \eqref{eq:tstat:update}, respectively, the identity \eqref{eq:tangent:identity} allows for recursive updating of the filter derivatives as well. Still, since neither the predictors nor the backward statistics are generally available in a closed form, we need to resort to approximations. This will be the topic of the next section.


\subsection{Particle-based smoothing in SSMs}  
\label{sec:partMethod}

\jimmy{Clearly, the approach to filter derivative estimation outlined above requires access to the filter and predictor distributions as well as the backward statistics. However, these quantities are available in a closed form only in very exceptional cases, e.g. when the state-space model is linear Gaussian or the state space $\set{X}$ is a finite set. In the general case we have to rely on approximations, and in this paper we employ SMC methods for this purpose.} As mentioned in the introduction, SMC methods are genetic-type algorithms propagating recursively a random sample of particles with associated weights through repeated importance sampling. 

In the following we assume that all random variables are well defined on a common probability space $(\Omega, \mathcal{F}, \prob)$.  

\subsubsection{The bootstrap particle filter} 
\label{sub:the_bootstrap_particle_filter}
The bootstrap particle filter~\citep{gordon:salmond:smith:1993} updates sequentially in time a set of particles with associated importance weights in order to approximate the filter and prediction distribution flows $\{\post{t; \parvec}\}_{t \in \nset}$ and $\{\pred{t; \parvec} \}_{t \in \nset}$, respectively, given the sequence $\{y_t\}_{t \in \nset}$ of observations. These methods are best illustrated in a recursive manner. Thus, assume that we have at hand a particle sample $\{\epart{t}{i}\}_{i=1}^{\N}$ approximating the predictor $\pred{t;\parvec}$ in the sense that for all $f \in \bmf{\alg{X}}$, as $\N$ tends to infinity, 
$$
	\partpred{t; \parvec} f \eqdef \frac{1}{\N} \sum_{i=1}^{\N} f(\epart{t}{i}) \backsimeq \pred{t; \parvec} f.
$$
Using the \emph{updating step}~\eqref{eq:pred:to:filt} of the filtering recursion we can, as soon as the observation $y_t$ becomes available, transform the uniformly weighted particle sample $\{\epart{t}{i}\}_{i=1}^{\N}$ into a weighted particle sample approximating $\post{t;\parvec}$ by associating each particle in the sample with the importance weight $\wgt{t}{i} \eqdef \md{t;\parvec}(\epart{t}{i})$. Now, the weighted particle sample $\{(\wgt{t}{i}, \epart{t}{i})\}_{i=1}^{\N}$ targets $\post{t;\parvec}$ in the sense that for all $f \in \bmf{\alg{X}}$, as $\N$ tends to infinity,
\begin{align}
	\post[part]{t; \parvec} f \eqdef \sum_{i=1}^{\N} \frac{\wgt{t}{i}}{\wgtsum{t}} f(\epart{t}{i}) \backsimeq \post{t; \parvec}f,
\end{align}
where $\wgtsum{t} \eqdef \sum_{i=1}^{\N} \wgt{t}{i}$ denotes the weight sum. In order to form a uniformly weighted particle sample $\{ \epart{t+1}{i}\}_{i=1}^{\N}$ targeting the next predictor $\pred{t+1; \parvec}$ we plug in the filter approximation $\post[part]{t; \parvec}$ into the \emph{prediction step} \eqref{eq:filt:to:pred}, which results in an approximation of $\pred{t+1;\parvec}$ given by a mixture proportional to $\sum_{i=1}^{\N} \wgt{t}{i} \hk_{\parvec}(\epart{t}{i},\cdot)$. The particle cloud is then updated through sampling from this mixture, an operation that is typically carried through in two steps, \emph{selection} and \emph{mutation}, which are analogous to updating and prediction. In the selection step, a set $\{I_{t+1}^i\}_{i=1}^{\N}$ of indices are drawn from $\probdist(\{\wgt{t}{i}\}_{i=1}^{\N})$, where for any positive numbers $\{a_i\}_{i = 1}^k$, $\probdist(\{a_i\}_{i = 1}^k)$ denotes the categorical distribution on $\{1, \ldots, k \}$ induced by the probabilities proportional to $\{a_i \}_{i = 1}^k$. After this, the mutation step propagates the particles forwards according to the dynamics of the state process, i.e., for all $i \in \{1,\ldots,\N\}$, 
\begin{align}
	\epart{t+1}{i} \sim \hk_{\parvec}(\epart{t}{I_{t+1}^i}, \cdot).
\end{align}
The algorithm is initialised by drawing $\{Ê\epart{0}{i} \}_{i = 1}^\N \sim \Xinit^{\varotimes \N}$. 

The bootstrap particle filter is well suited for approximating the filter and prediction distribution flows, but it can also be used for estimating the flow of joint-smoothing distributions. In the \emph{Poor man's smoother}, this is done by tracing, backwards in time, the genealogical history of the particles, and using the empirical measures associated with these ancestral lineages as approximations of the joint-smoothing distributions. Unfortunately, the repeated resampling operations of bootstrap particle filter collapses these trajectories in the long run, implying, for large $t$, the existence of a random time point $T < t$ before which all the trajectories coincide, i.e., all the particles $\{ \epart{t}{i} \}_{i = 1}^\N$ have a common time $T$ ancestor. Thus, this naive approach leads to a severely depleted estimator; {see~\cite{kitagawa:1996,kitagawa:sato:2001} for discussions of this particle path degeneracy phenomenon and \cite{jacob:murray:rubenthaler:2013,olsson:cappe:douc:moulines:2006,olsson:westerborn:2014b} for theoretical analyses of the same.} 

\subsection{Online smoothing of additive functionals} 
\label{sub:smoothing_of_additive_functionals}

A way of detouring the particle path degeneracy phenomenon goes via the backward decomposition~\eqref{eq:bk:decomp}. Moreover, as we will see next, \eqref{eq:tstat:update} provides, when the objective function of interest is of additive form, a means for recursive, non-collapsed particle smoothing. This requires each backward kernel $\bk{\post{s;\parvec}}$, $s \in \nset$, to be approximated, which is naturally done through 
\begin{equation} \label{eq:part:backward}
	\bk{\post[part]{s;\parvec}}f(x) = \sum_{i=1}^{\N} \frac{\wgt{s}{i} \hd[\parvec](\epart{s}{i},x)}{\sum_{\ell =1}^{\N} \wgt{s}{\ell} \hd[\parvec](\epart{t}{\ell},x)}f(\epart{s}{i}), \quad (x, f) \in \set{X} \times \bmf{\alg{X}}.  \label{eq:bk:part}
\end{equation}
The forward-only smoothing algorithm proposed by \cite{delmoral:doucet:singh:2010} consists in replacing, in \eqref{eq:tstat:update}, $\bk{\post{t;\parvec}}$ by the approximation $\bk{\post[part]{s;\parvec}}$. Proceeding again recursively and assuming that we have at hand estimates $\{\tstattil{i}{t}\}_{i=1}^{\N}$ of $\{ \tstat{t}\af{t}(\epart{t}{i})\}_{i=1}^{\N}$, we obtain the updating formula
\begin{equation}
	\tstattil{i}{t+1} = \sum_{j = 1}^{\N} \frac{\wgt{t}{j} \hd[\parvec](\epart{t}{j}, \epart{t+1}{i})}{\sum_{\ell = 1}^{\N} \wgt{t}{\ell} \hd[\parvec](\epart{t}{\ell}, \epart{t+1}{i}) }\left( \tstattil{j}{t} + \addf{t}(\epart{t}{j}, \epart{t+1}{i}) \right). \label{eq:update:FFBSm}
\end{equation}
The recursion is initialised by setting $\tstattil{i}{0} = 0$ for all $i \in \{1,\ldots,\N\}$. On the basis of these estimates, each expectation $\post{0:t \mid t - 1; \parvec}\af{t}$ may then be approximated by
$$
	\post[part]{0:t \mid t - 1;\parvec} \af{t} \eqdef \frac{1}{\N} \sum_{i=1}^{\N} \tstattil{i}{t}.
$$
This method allows for online estimation of $\post{0:t+1\mid t;\parvec}\af{t+1}$. It also has the appealing property that only the current statistics $\{\tstattil{i}{t}\}_{i=1}^{\N}$ and particle cloud $\{(\epart{t}{i}, \wgt{t}{i})\}_{i=1}^{\N}$ need to be stored. However, since the updating formula \eqref{eq:update:FFBSm} requires, at each time step, a sum of $\N$ terms to be computed for each particle, the resulting algorithm has a computational complexity that grows \emph{quadratically} with the number $\N$ of particles. \jimmy{Needless to say, this forces the user to keep the particle sample size relatively small, resulting in low precision.}

\subsubsection{The PaRIS}

In order to increase the accuracy for a given computational budget, \cite{olsson:westerborn:2014b} replace \eqref{eq:update:FFBSm} by the Monte Carlo estimate
\begin{equation}
	\tstat[i]{t+1} = \frac{1}{\K} \sum_{j=1}^{\K} \left( \tstat[\bi{t+1}{i}{j}]{t} + \addf{t}(\epart{t}{\bi{t+1}{i}{j}}, \epart{t+1}{i}) \right), \label{eq:update:paris}
\end{equation}
where the sample size $\K \in \nsetpos$ (the so-called \emph{precision parameter}) is typically very small compared to $\N$ and $\{ \bi{t+1}{i}{j} \}_{j=1}^{\K}$ are i.i.d. samples from $\probdist(\{\wgt{t}{\ell} \hd[\parvec](\epart{t}{\ell}, \epart{t+1}{i}) \}_{\ell=1}^{\N})$. Using the updated $\{\tstat[i]{t+1}\}_{i=1}^{\N}$, an estimate of $\post{0:t+1 \mid t;\parvec} \af{t+1} = \pred{t + 1 ; \parvec} \tstat{t+1;\parvec} \af{t+1}$ is obtained as $\N^{-1} \sum_{i=1}^{\N} \tstat[i]{t+1}$. As before, the algorithm is initialised by setting $\tstat[i]{0} = 0$ for $i \in \{1, \ldots, \N\}$. 

In its most basic form, also the previous approach has an $O(\N^2)$ complexity since it requires the normalising constant $\sum_{\ell = 1}^\N \wgt{t}{\ell} \hd[\parvec](\epart{t}{\ell}, \epart{t+1}{i})$ of the distribution $\probdist(\{\wgt{t}{\ell} \hd[\parvec](\epart{t}{\ell}, \epart{t+1}{i}) \}_{\ell=1}^{\N})$ to be computed particle-wise, i.e. for all $\epart{t+1}{i}$. In order to speed up the algorithm we apply the accept-reject sampling approach proposed by~\cite{douc:garivier:moulines:olsson:2010}. This technique presupposes that there exists a constant $\hkup > 0$ such that $\hd[\parvec](x,x') \leq \hkup$ for all $(x, x') \in \set{X}^2$, an assumption that is satisfied for most models. Then, in order to sample from $\probdist(\{\wgt{t}{\ell} \hd[\parvec](\epart{t}{\ell}, \epart{t+1}{i}) \}_{\ell=1}^{\N})$ a candidate $J^\ast \sim \probdist(\{\wgt{t}{i}\}_{i=1}^{\N})$ is accepted with probability $\hd[\parvec](\epart{t}{J^*}, \epart{t+1}{i})/\hkup$. This procedure is repeated until acceptance. Under certain mixing assumptions (see \autoref{ass:mixing} below) it can be shown (see~\citet[Proposition 2]{douc:garivier:moulines:olsson:2010} and \citet[Theorem~10]{olsson:westerborn:2014b}) that the expected number of trials needed for this approach, which is referred to as the \emph{particle-based, rapid incremental smoother} (PaRIS), to update $\{\tstat[i]{t}\}_{i=1}^{\N}$ to $\{\tstat[i]{t+1}\}_{i=1}^{\N}$ is $O(\K \N)$.

A key discovery of~\cite{olsson:westerborn:2014b} is that the algorithm converges, as $\N$ tends to infinity, for all fixed $\K \in \nsetpos$. In particular, \citet[Corollary~5]{olsson:westerborn:2014b} provide a CLT with an asymptotic variance that may, as long as $\K \geq 2$, be shown to grow only \emph{linearly} with $t$, which is optimal for a Monte Carlo estimator on the product space \citep[we refer to][Section~1, for a discussion]{olsson:westerborn:2014b}. In addition, the additional asymptotic variance imposed by the ``de-Rao-Blackwellisation'' strategy \eqref{eq:update:paris} is inversely proportional to $\K - 1$, which suggests that $\K$ should be kept at a modest value (in fact, using, say, $\K \in \{2, 3\}$ works typically well in simulations). On the other hand, in the case $\K = 1$, the estimator suffers from a degeneracy problem that resembles closely that of the Poor man's smoother, implying an asymptotic variance that grows quadratically fast with $t$. Similar stochastic stability aspects of the tangent filter estimator proposed in the present paper will be discussed in \autoref{sec:convergence}.


\section{Tangent filter estimation: main results}
\label{sec:main:results}

\subsection{PaRIS-based estimation of tangent filters}
\label{sec:algo}
{The estimator that we propose is based on the identity \eqref{eq:tangent:identity}, which we for a given $\parvec$ approximate online using the PaRIS algorithm. More specifically, we estimate, on-the-fly as new observations appear, the flow $\{\filtderiv{t; \parvec}\}_{t \in \nset}$ of tangent filters. This is done by running a particle filter targeting the predictor distributions while updating estimates $\{\tstat[i]{t}\}_{i=1}^{\N}$ of $\{\tstat{t;\parvec}\af{t;\parvec}(\epart{t}{i})\}_{i = 1}^{\N}$, where $\af{t; \parvec}$ is the complete data score given by \eqref{eq:complete:score}. The statistics $\{\tstat[i]{t}\}_{i=1}^{\N}$ are updated efficiently using the PaRIS updating rule~\eqref{eq:update:paris}. At each time-step $t$, our algorithm returns}
\begin{align}
	\filtderiv[part]{t; \parvec} \testf[t] = \frac{1}{\N} \sum_{i=1}^{\N} \left( \tstat[i]{t} - \frac{1}{\N}\sum_{j=1}^{\N} \tstat[j]{t} \right) \testf[t](\epart{t}{i})
\end{align}
as an estimate of $\filtderiv{t; \parvec} \testf[t]$ for any $\testf[t] \in \bmf{\alg{X}}$. We summarise the algorithm in the pseudo-code below, where every line with $i$ should be executed for all $i \in \{1,\ldots,\N\}$. \smallskip

\begin{algorithmic}[1]
\State draw $\epart{0}{i} \sim \Xinit$
\State set $\tstat[i]{0} \gets 0$
\For{$t \gets 0, 1, 2, \ldots$}
\State set $\wgt{t}{i} \gets \md{t;\parvec}(\epart{t}{i})$
\State draw $I_{t + 1}^i \sim \probdist(\{\wgt{t}{\ell}\}_{\ell=1}^{\N})$
\State draw $\epart{t + 1}{i} \sim \hd[{\parvec}](\epart{t}{I^i_{t + 1}}, \cdot)$
\For{$\k \gets 1, \ldots, \K$}
\State draw $\bi{t + 1}{i}{\k} \sim \probdist( \{ \wgt{t}{\ell} \hd[{\parvec}](\epart{t}{\ell}, \epart{t + 1}{i} ) \}_{\ell = 1}^\N)$
\EndFor
\State set 
$
\tstat[i]{t + 1} \gets \frac{1}{\K} \sum_{\k = 1}^{\K}\left( \tstat[\bi{t + 1}{i}{\k}]{t} + \addf{t}(\epart{t}{\bi{t + 1}{i}{\k}}, \epart{t + 1}{i}) \right)
$
\State set $\bar{\tau}_{t + 1} \gets \frac{1}{\N} \sum_{\ell = 1}^{\N} \tstat[\ell]{t + 1}$
\State set $\filtderiv[part]{t+1 ; \parvec} \gets \frac{1}{\N} \sum_{\ell = 1}^{\N} \left( \tstat[\ell]{t+1} - \bar{\tau}_{t+1} \right) \delta_{\epart{t+1}{\ell}}$
\EndFor
\end{algorithmic}

\begin{remark}
The previous algorithm can be viewed as a ``de-Rao-Blackwellisation'' of the tangent filter estimator proposed in \cite{delmoral:doucet:singh:2015}, which is based on the updating rule \eqref{eq:update:FFBSm} instead of \eqref{eq:update:paris}. With the exception of this important difference, the two algorithms follow the same recursion.  Consequently, the algorithm of \cite{delmoral:doucet:singh:2015} has an $O(\N^2)$ complexity.   
\end{remark}


\subsection{Theoretical results} 
\label{sec:convergence}
The first part of the convergence analysis of our algorithm will be carried through under the following, relatively mild, assumption.  
\begin{assumption} \label{ass:bounded:q:g}\

	\begin{enumerate}[(i)]
		\item For all $t \in\nset$ and $\parvec \in \Theta$, $\md{t;\parvec}$ is a positive and bounded measurable function.
		\item For all $\parvec \in \Theta$, $\hd[\parvec] \in \bmf{\alg{X}^{\varotimes 2}}$. 
	\end{enumerate}
\end{assumption}

\autoref{ass:bounded:q:g}(i) implies finiteness and positiveness of the particle importance weights, while \autoref{ass:bounded:q:g}(ii) allows, apart from some technical arguments needed in the proofs, the accept-reject sampling technique mentioned in \autoref{sub:smoothing_of_additive_functionals} to be used. 

We begin by establishing an exponential concentration inequality, valid for all finite particle sample sizes $\N$, for our PaRIS-based tangent filter estimator. This yields the strong ($\prob$-a.s.) consistency (as $\N$ tends to infinity) of the algorithm as a corollary. In addition, we state and establish a CLT with an asymptotic variance given by a sum of the asymptotic variance of the estimator proposed by~\cite{delmoral:doucet:singh:2015} and one term reflecting the additional variance introduced by the supplementary sampling performed in the PaRIS-version. Finally, we are able to derive, by operating with strong mixing assumptions, a time uniform bound on the asymptotic variance, which is inversely proportional to $\K - 1$. This guarantees that the algorithm is stochastically stable in the long run. All proofs, which use results obtained by \cite{olsson:westerborn:2014b}, are postponed to~\autoref{sec:proofs}. 

\subsubsection{Results on convergence}

\begin{theorem} \label{thm:hoeffding}
	Let~\autoref{ass:bounded:q:g} hold. Then for all $t \in \nset$, $\parvec \in \Theta$, $\testf[t] \in \bmf{\alg{X}}$, and $\K \in \nsetpos$ there exists $(c_t, \tilde{c}_t) \in (\mathbb{R}_+^\ast)^2$ (depending on $\theta$, $\K$, and $\testf[t]$) such that for all $\N \in \nsetpos$ and $\varepsilon \in \mathbb{R}_+$,
	$$
		\prob \left( \left| \filtderiv[part]{t; \parvec} \testf[t] - \filtderiv{t;\parvec} \testf[t] \right| \geq \varepsilon \right) \leq c_t \exp \left( - \tilde{c}_t \N \varepsilon (\varepsilon \wedge 1) \right).
	$$
\end{theorem}

The strong consistency of the estimator follows. 
 
\begin{corollary} \label{corollary:as:conv}
	Let~\autoref{ass:bounded:q:g} hold. Then for all $t \in \nset$, $\parvec \in \Theta$, $\testf[t] \in \bmf{\alg{X}}$, and $\K \in \nsetpos$ it holds, $\prob$-a.s.,  
	$$
		\lim_{\N \to \infty} \filtderiv[part]{t; \parvec} \testf[t] = \filtderiv{t; \parvec} \testf[t].
	$$
\end{corollary}

We now set focus on weak convergence. For this purpose, we introduce, for all $t \in \nset$, the unnormalised kernels 
$$
\lk{t;\parvec} \testf (x) \eqdef \md{t;\parvec}(x) \hk_{\parvec} \testf(x), \quad (x, \testf) \in \set{X} \times \bmf{\alg{X}},
$$ 
and, for all $(s,t) \in \nset^2$ such that $s \leq t$, the retro-prospective kernels 
\[
	\begin{split}
		\BF{s}{t;\parvec}\testf(x_{s}) & \eqdef \int \int \testf(x_{0:t}) \md{t;\parvec}(x_t) \, \tstat{s;\parvec}(x_s, \rmd x_{0:s-1}) \, \lk{s;\parvec} \varotimes \cdots \varotimes \lk{t-1;\parvec} (x_s, \rmd x_{s+1:t}) , \\
		\BFcent{s}{t;\parvec} \testf(x_{s}) & \eqdef \BF{s}{t;\parvec}(\testf - \post{0:t \mid t-1;\parvec} \testf)(x_s),  
	\end{split}
\]
for $(x_s, \testf) \in \set{X} \times \bmf{\alg{X}^{\varotimes (t + 1)}}$, operating simultaneously in the backward and forward directions.

\begin{theorem} \label{thm:CLT}
	Let~\autoref{ass:bounded:q:g} hold. For all $t \in \nset$, $\parvec \in \Theta$, $\K \in \nsetpos$, and $\testf[t] \in \bmf{\alg{X}}$, as $\N \to \infty$,
	$$
		\sqrt{N} \left( \filtderiv[part]{t;\parvec}\testf[t] - \filtderiv{t;\parvec}\testf[t] \right) \convd \asvard{t;\parvec}{\testf[t]} Z,
	$$
	where $Z$ is a standard Gaussian random variable and 
	\begin{equation}
		\sigma_{t;\parvec}^2(\testf[t]) \eqdef \tilde{\sigma}_{t;\parvec}^2(\testf[t]) + \sum_{s=0}^{t - 1} \sum_{\ell=0}^{s} \K^{\ell - (s+1)} \asvardPaRIS{s}{\ell}{t}{\testf[t]}, \label{eq:asvar:short}
	\end{equation}	
	where 
	\begin{equation} \label{eq:as:var:del:moral}
		\tilde{\sigma}_{t;\parvec}^2(\testf[t]) \eqdef \sum_{s=0}^{t-1} \frac{\pred{s+1;\parvec} \BFcent[2]{s+1}{t;\parvec} \{ (\af{t;\parvec} - \pred{t;\parvec} \tstat{t;\parvec} \af{t;\parvec})(\testf[t] - \pred{t;\parvec} \testf[t])\}}{(\pred{s+1;\parvec} \lk{s+1;\parvec} \cdots \lk{t-1;\parvec} \1{\set{X}})^2} 
	\end{equation}
	is the asymptotic variance of the tangent filter estimator proposed by~\cite{delmoral:doucet:singh:2015} and 
	\begin{multline}
		\asvardPaRIS{s}{\ell}{t}{\testf[t]} \eqdef \\
		\frac{\pred{\ell+1;\parvec} \{ \bk{\post{\ell;\parvec}}(\tstat{\ell;\parvec} \af{\ell;\parvec} + \addf{\ell;\parvec} - \tstat{\ell+1;\parvec} \af{\ell + 1;\parvec})^2 \lk{\ell + 1;\parvec} \cdots \lk{s;\parvec}( \lk{s+1;\parvec} \cdots \lk{t-1;\parvec} \{\testf[t] - \pred{t;\parvec} \testf[t]\})^2\} }{(\pred{\ell+1;\parvec} \lk{\ell + 1;\parvec} \cdots \lk{s;\parvec} \1{\set{X}})(\pred{s+1;\parvec} \lk{s+1;\parvec} \cdots \lk{t-1;\parvec} \1{\set{X}})^2}. 
	\end{multline}
 \end{theorem}

\begin{remark}
	We note that for all $t \in \nset$, 
	\begin{align}
		\lim_{\K \to \infty} \asvard[2]{t;\parvec}{\testf[t]} = \tilde{\sigma}_{t;\parvec}^2(\testf[t]). 
	\end{align}
	This is in line with our expectations, as the tangent filter estimator proposed by~\cite{delmoral:doucet:singh:2015} can be viewed as a Rao-Blackwellisation of our estimator.  
\end{remark}

\subsubsection{Results on stochastic stability}

We now set focus on the numerical stability of the algorithm, and show that the error of the tangent filter estimator proposed by us stays uniformly bounded in time. There are several ways of formulating such stability, and in the present paper we follow \cite{delmoral:guionnet:2001,douc:moulines:olsson:2014} and establish a time-uniform bound on the asymptotic variance of the estimator. The analysis requires \autoref{ass:bounded:q:g} to be sharpened as follows.

\begin{assumption} \label{ass:mixing} \
	\begin{itemize}
		\item[(i)] There exists $\rho > 1$ such that for all $\parvec \in \Theta$ and $(x, \tilde{x}) \in \set{X}^2$,
		$$
		 	\rho^{-1} \leq \hd[\parvec](x, x') \leq \rho.
		 $$ 
		 \item[(ii)] There exist constants $0 < \mdlow < \mdup < \infty$ such that for all $t \in \nset$ all $\parvec \in \Theta$, and $x \in \set{X}$, 
		 $$ 
			 \mdlow \leq \md{t;\parvec}(x) \leq \mdup. 
		 $$
	\end{itemize}
\end{assumption}
	 
Under \autoref{ass:mixing} we define
$$
\mr{} \eqdef 1 - \rho^{-2}.
$$

These---admittedly very restrictive---\emph{strong mixing assumptions} require typically the state space $\set{X}$ of the underlying Markov chain to be a compact set. Nevertheless, such assumptions are standard in the literature of SMC analysis; see~\cite{delmoral:guionnet:2001} and, e.g., \cite{crisan:heine:2008,douc:moulines:olsson:2014} for refinements. Finally, we assume that the terms \eqref{eq:complete:score} of the complete data score function are uniformly bounded. 

\begin{assumption} \label{ass:bdd} 
	There exists a constant $\hbd \in \mathbb{R}_{+}$ such that for all $s \in \nset$ and $\parvec \in \Theta$, 
	$$
	\oscn{\addf{s;\parvec}} \leq \hbd.
	$$ 
\end{assumption}
{In conformity with the strong mixing assumptions, also  \autoref{ass:bdd} requires typically the state space $\set{X}$ to be compact.} 

The previous assumptions allow us to establish the following theorem, providing a time-uniform bound on the sequence $\{ \sigma_{t;\parvec}^2(\testf) \}_{t \in \nset}$ {of asymptotic variances}. {Generally, when the objective function is of additive form, the variance of a Monte Carlo estimator increases typically linearly with the number of terms, or, the dimension of the underlying product space \citep[see e.g.][Section~1, for a discussion]{olsson:westerborn:2014b}. Thus, as our estimator is based on the identities \eqref{eq:tangent:identity:first:step} and \eqref{eq:tangent:identity}, which express the tangent filter in terms of additive smoothed expectations over the path space, it may seem surprising that we are able to derive a time-uniform bound on the asymptotic variance. However, 
recall from \eqref{eq:tangent:identity:first:step} that the tangent filter $\filtderiv{t; \parvec} f_t$ coincides with the covariance of $f_t$ and the complete data score function $\af{t;\parvec}$ under $\post{0:t \mid t - 1; \parvec}$. Thus, when the model exhibits forgetting properties---which is the case under the strong mixing assumptions above---the sequence $\{ \filtderiv{t; \parvec} f_t \}_{t \in \nset}$ stays bounded, and time uniform convergence of the Monte Carlo estimator is hence within reach.}

 \begin{theorem}\label{thm:bdd:var}
	Let~\autoref{ass:mixing} and~\autoref{ass:bdd} hold. Then there exist constants $(c, \tilde{c}) \in \rset_+^2$ such that for all $t \in \nset$, $\theta \in \Theta$, $\K \geq 2$, and $\testf[t] \in \bmf{\alg{X}}$, 
	\begin{equation}
		\asvard[2]{t;\parvec}{\testf[t]} \leq \| \testf[t] \|_\infty^2 \left( c + \tilde{c} \frac{1}{(\K - 1)(1 - \mr^2)} \right) \label{eq:bound:asvar}
	\end{equation}
	and 
	\begin{equation} \label{FFBSm:var:bound}
		\tilde{\sigma}_{t;\parvec}^2(\testf[t]) \leq c \| \testf[t] \|_\infty^2. 
	\end{equation}
\end{theorem}
The bound \eqref{FFBSm:var:bound} on the first term \eqref{eq:as:var:del:moral} of the asymptotic variance, corresponding to the variance of the Rao-Blackwellised algorithm, was derived by \cite{delmoral:doucet:singh:2015}. As expected from the theoretical results obtained by \cite{olsson:westerborn:2014b}, the bound on the second term is inversely proportional to $\K - 1$, which indicates that boosting the number $\K$ of PaRISian backward draws from a moderate to a very large value will have negligible effect on the precision. We hence recommend keeping $\K$ low, say lower than $5$, in order to gain computational efficiency. A more detailed discussion on the choice of $\K$ can be found in~\cite{olsson:westerborn:2014b}.



\section{Application to recursive maximum likelihood estimation}
\label{sec:par:est}

\subsection{Batch mode implementation}

Given a fixed data-record $y_{0:t}$, our goal is to perform maximum likelihood estimation of $\parvec$, i.e., to find the vector of parameters $\parvec^{*} \in \parspace$ such that $\parvec^{*} = \arg\max_{\parvec \in \parspace} \llh{\parvec}(y_{0:t})$, where $\llh{\parvec}(y_{0:t})$ is the likelihood of the observed data given in~\eqref{eq:llh}. Equivalently, we may maximise the log-likelihood $\logllh{\parvec}(y_{0:t}) = \log \llh{\parvec}(y_{0:t})$. There are many different approaches to such maximisation; see~\cite[Chapter 10]{cappe:moulines:ryden:2005} for a more general overview of different approaches.

Here we will focus on the following \emph{Robbins-Monro scheme}: at iteration $n$, let 
$$
\parvec_n = \parvec_{n - 1} + \gamma_n Z_n, 
$$ 
where $Z_n$ is a noisy measurement of $\deriv \logllh{\parvec}(y_{0:t}) |_{\parvec = \parvec_{n-1}}$, i.e., the score of the observed data evaluated in $\parvec = \parvec_{n-1}$, and $\{\gamma_n\}_{n \in \nsetpos}$ a sequence of positive step-sizes satisfying the regular stochastic approximation requirements $\sum_{n = 1}^\infty \gamma_n = \infty$ and $\sum_{n = 1}^\infty \gamma_n^2 < \infty$. Note that this approach requires approximation of $Z_n$ at each iteration of the algorithm. If the number of observations is very large computing $Z_n$ is costly, and since many iterations are often required for convergence this results in an impractical algorithm. Moreover, if we receive a new observation $Z_n$ needs to be recalculated which turns the procedure into an offline algorithm.

\subsection{Online implementation: PaRISian RML}
\label{sec:rml:paris}

We sketch the basic principles of recursive maximum likelihood. {First note that since $\{ (X_t, Y_t) \}_{t \in \nset}$ is a Markov chain, the {quadruple} $\{ (X_t, Y_t, \pred{t;\parvec}, \filtderiv{t; \parvec}) \}_{t \in \nset}$ forms, by \eqref{eq:pred:to:filt} and \eqref{eq:filt:to:pred} and for all $\parvec \in \Theta$, a Markov chain as well. In the case where the state space $\set{X}$ is a finite set, \cite{legland:mevel:1996} showed that this chain is ergodic under certain mixing assumptions.} This result was later extended by \cite{douc:matias:2002} to the case where $\set{X}$ is compact. Now, assume that data is generated by a model parameterised by a {distinctive} $\truepar \in \Theta$; then, denoting by $\stat_{\parvec, \truepar}$ the stationary distribution of this chain and $\tilde{\stat}_{\parvec, \truepar}(\cdot) = \stat_{\parvec, \truepar}(\set{X} \times \cdot)$ its marginal with respect to the last three components, it holds for all $\parvec \in \Theta$, $\prob$-a.s., 
\begin{multline} \label{eq:erg:limit}
\lim_{t \to \infty} \frac{1}{t} \deriv \logllh{\parvec}(Y_{0:t}) = \lim_{t \to \infty} \frac{1}{t} \sum_{s = 0}^t \deriv \logllh{\parvec}(Y_s \mid Y_{0:s - 1})  
= \lim_{t \to \infty} \frac{1}{t} \sum_{s = 0}^t \frac{\pred{s; \parvec}(\deriv \md{s;\parvec}) + \filtderiv{s; \parvec} \md{s ;\parvec}}{\pred{s; \parvec} \md{s; \parvec}} \\
 = \iiint \frac{\pi(\deriv \md{\parvec}(\cdot, y)) + \eta (\md{\parvec}(\cdot, y))}{\pi (\md{\parvec}(\cdot, y))} \, \tilde{\stat}_{\parvec, \truepar}(\rmd (y, \pi, \eta)) \defeq \lambda(\parvec, \truepar),  
\end{multline}
where $\pred{s; \parvec} \md{s ; \parvec}$, $\filtderiv{s; \parvec} \md{s ;\parvec}$, and $\pred{s; \parvec} \md{s; \parvec}$ depend implicitly on $Y_{0:s}$. {These equations follow, in order, by the decomposition \eqref{eq:log-likelihood:sum} of the log-likelihood, by applying the nabla operator, and by applying Birkhoff's ergodic theorem.} By indentifiability \cite[see, e.g.,][Section~12.4]{cappe:moulines:ryden:2005}, the true parameter $\truepar$ solves $\lambda(\parvec, \truepar) = 0$, and the task of solving this equation may be cast into the framework of \emph{stochastic approximation} and the \emph{Robbins-Monro algorithm}
\begin{equation}
	\parvec_{t+1} = \parvec_t + \gamma_{t+1} \zeta_{t+1}, \quad t \in \nset, 
	\label{eq:par:update}
\end{equation}
where $\zeta_{t+1}$ is a noisy observation of $\lambda(\parvec_t, \truepar)$. Ideally, such a noisy observation would be formed by estimating $\lambda(\parvec_t, \truepar)$ 
by a draw from $\tilde{\stat}_{\parvec, \truepar}$, which is, needless to say, not possible in practice. Thus, by introducing one more level of approximation and allowing for Markovian perturbations, one may instead estimate, following \eqref{eq:tangent:identity}, $\lambda(\parvec_t, \truepar)$ by 
\begin{equation} \label{eq:def:zeta}
	\zeta_{t + 1} \eqdef \frac{\zeta^{1}_{t+1} + \zeta^{2}_{t+1}}{\zeta^{3}_{t+1}},
\end{equation}
where
\begin{equation} \label{eq:zeta:terms}
\begin{split}
	\zeta^{1}_{t+1} &\eqdef \pred{t+1} \left( \deriv \md{t+1;\parvec} |_{\parvec = \parvec_t} \right),  \\
	\zeta^{2}_{t+1} &\eqdef \filtderiv{t+1} \md{t + 1;\parvec_t}, \\
	\zeta^{3}_{t+1} &\eqdef \pred{t+1} \md{t + 1;\parvec_t},
\end{split}
\end{equation}
with 
$$
   \filtderiv{t+1} \eqdef \pred{t}\{ (T_t - \pred{t} T_t) \md{t;\parvec_t} \},
$$
and statistics $\{ T_t \}_{t \in \nset}$ being, following \eqref{eq:tstat:update}, updated recursively according to  
\begin{equation} \label{eq:tstat:RML:update}
	T_{t + 1}(x_{t + 1}) = \int \{ T_t(x_t) + \addf{t;\parvec_{t + 1}} (x_{t:t + 1}) \} \, \bk{\post{t}}(x_{t + 1}, \rmd x_t), \quad x_{t + 1} \in \set{X},  
\end{equation}
{and initialised as $T_0 \equiv 0$. The statistic \eqref{eq:tstat:RML:update}} depends on the new observation $Y_{t + 1}$ through the term  $\addf{t;\parvec_{t + 1}}$. In addition, the measures $\{ \post{t} \}_{t \in \nset}$ and $\{ \pred{t} \}_{t \in \nset}$ are updated according to the recursion 
\begin{equation} \label{eq:mod:filter:rec}
\begin{split}
	\post{t} \testf&= \frac{\pred{t}(\md{t;\parvec_t} \testf)}{\pred{t} \md{t;\parvec_t}},\\
	\pred{t + 1} \testf &= \post{t} \hk_{\parvec_t} \testf, 
\end{split}
\quad f \in \bmf{\alg{X}}, 
\end{equation}
with, by convention, $\pred{0} \eqdef \Xinit$. {Note that the recursion  \eqref{eq:mod:filter:rec} updates the filter and predictor distributions on the basis of the current parameter fit. This means that for all $t \in \nset$, the measures $\phi_t$, $\pi_{t + 1}$ and $\eta_{t + 1}$ depend on \emph{all} the past parameter values $\{ \parvec_s \}_{s = 0}^t$.} 
  
This approach yields an online algorithm where a new observation can be incorporated into the estimator without the need of having to recalculate the latter from scratch. In the case where $\set{X}$ is finite, this algorithm was studied in~\cite{legland:mevel:1997}, which also provides assumptions under which the sequence $\{\parvec_{t}\}_{t \in \nsetpos}$ converges towards the parameter $\parvec_\ast$ \cite[see also][for refinements]{Tadic:2010}.

In the general case, the recursions \eqref{eq:tstat:RML:update} and \eqref{eq:mod:filter:rec} are intractable. The translation of the same into the language of particles is however immediate by approximating \eqref{eq:mod:filter:rec} by a particle filter, updated according to time-varying parameters, and  \eqref{eq:tstat:RML:update} by means of a PaRISian updating step \eqref{eq:update:paris}, the latter executed for some suitable precision parameter $\K \geq 2$.   

The algorithm is initialised by some arbitrary parameter guess $\parvec_0$. After this, the parameter fit is updated recursively as each new observation $y_{t + 1}$ is received, by first calculating a particle approximation $\hat{\zeta}_{t + 1}$ of $\zeta_{t + 1}$ and then updating the parameter according to
\begin{equation}
	\parvec_{t+1} = \parvec_t + \gamma_{t+1} \hat{\zeta}_{t+1},  
\end{equation}
where $\{ \gamma_t \}_{t \in \nsetpos}$ satisfies the standard stochastic optimisation requirements (i.e. infinite sum and finite sum of squares) above. The algorithm, which is illustrated numerically in the next section, is detailed in pseudo code below. \smallskip

\begin{algorithmic}[1]
\State set arbitrarily $\parvec_0$
\State draw $\epart{0}{i} \sim \Xinit$
\State set $\tstat[i]{0} \gets 0$
\For{$t \gets 0, 1, 2, \ldots$}
\State set $\wgt{t}{i} \gets \md{\parvec_t}(\epart{t}{i}, y_t)$
\State draw $I_i \sim \probdist(\{\wgt{t}{\ell}\}_{\ell=1}^{\N})$
\State draw $\epart{t + 1}{i} \sim \hd[{\parvec_t}](\epart{t}{I_i}, \cdot)$
\For{$\k \gets 1, \ldots, \K$}
\State draw $\bi{t + 1}{i}{\k} \sim \probdist( \{ \wgt{t}{\ell} \hd[{\parvec_t}](\epart{t}{\ell}, \epart{t + 1}{i} ) \}_{\ell = 1}^\N)$
\EndFor
\State set 
$
\tstat[i]{t + 1} \gets \K^{-1}\sum_{\k = 1}^{\K} \left( \tstat[\bi{t + 1}{i}{\k}]{t} + \addf{t;\parvec_t}(\epart{t}{\bi{t + 1}{i}{\k}}, \epart{t + 1}{i}) \right)
$
\State set $\bar{\tau}_{t + 1} \gets \N^{-1} \sum_{\ell = 1}^{\N} \tstat[\ell]{t + 1}$
\State set	$\rmlterm{t + 1}{1} \gets \N^{-1} \sum_{\ell = 1}^{\N} \nabla \md{\parvec_t}(\epart{t + 1}{\ell}, y_{t + 1})$
\State set $\rmlterm{t + 1}{2} \gets \N^{-1} \sum_{\ell = 1}^{\N} \left( \tstat[\ell]{t + 1} - \bar{\tau}_{t + 1} \right) \md{\parvec_t}(\epart{t + 1}{\ell}, y_{t + 1})$
\State set $\rmlterm{t + 1}{3} \gets \N^{-1} \sum_{\ell = 1}^{\N} \md{\parvec_t}(\epart{t + 1}{\ell}, y_{t + 1})$
\State set $\parvec_{t + 1} \gets \parvec_t + \gamma_{t + 1} \dfrac{\rmlterm{t + 1}{1} + \rmlterm{t + 1}{2}}{\rmlterm{t + 1}{3}}$
\EndFor
\end{algorithmic}

\subsection{Simulations} 
\label{sec:simulations}

We benchmark our algorithm on two different models, namely
\begin{itemize}
	\item the stochastic volatility model of~\cite{hull:white:1987}, where we compare our estimator to that of \cite{delmoral:doucet:singh:2015}, and
	\item the simultaneous localisation and mapping (SLAM) model where we apply our algorithm as a proof of concept.
\end{itemize}

\subsubsection{Stochastic volatility model}
 
Consider the stochastic volatility model
\[
    \begin{split}
        X_{t + 1} &= \phi X_t + \sigma V_{t + 1}, \\
        Y_t &= \beta \exp(X_t / 2) U_t,
    \end{split}
    \quad t \in \nset \eqsp,
\]
proposed by \cite{hull:white:1987}, where $\{V_t \}_{t \in \nset^*}$ and $\{U_t\}_{t \in \nset}$ are independent sequences of mutually independent standard Gaussian noise variables. In this model the parameters to be estimated are $\parvec = (\phi, \sigma^2, \beta^2)$. We compared our PaRISian~RML estimator to the estimator proposed by~\cite{delmoral:doucet:singh:2015}{---the latter being referred to as the ``Rao-Blackwellised estimator''.} Since there is a significant difference in computational complexity between the algorithms, we design the particle sample size in each algorithm in such a way that the running times of the two algorithms are close to identical. With our implementation, using $\N = 100$ in the algorithm of \cite{delmoral:doucet:singh:2015} yields close to the same CPU time as using $(\N, \K) = (1400, 2)$ in the PaRISian RML estimator. For both algorithms we let $\{\gamma_t \}_{t \in \nset} = \{t^{-0.6}\}_{t \in \nset}$ . The algorithms were executed on an observation data record comprising $500,\!000$ observations generated under the parameters $\parvec^{*} = (0.8, 0.1, 1)$. The algorithms were run $12$ times each on the same data set with the same randomised starting parameters. \autoref{fig:trajectories} displays the resulting learning trajectories. It can be seen clearly that both methods converge to the true parameter values; however, the PaRISian RML estimates exhibit significantly less variance than those produced by the Rao-Blackwellised estimator (using a considerably smaller particle sample size). Indeed, computing the sample variances over the final parameter estimates indicates an improvement by almost an order of magnitude to the benefit of the PaRISian version for the same computational time: $(.069, .181, .095) \times 10^{-4}$ and $(.054, .164, .063) \times 10^{-3}$ for the PaRISian and Rao-Blackwellised versions, respectively. It can also be seen in~\autoref{fig:trajectories} that the trajectories exhibit some jumps, mainly in the $\beta^2$ variable. This occurs when the estimate of $\zeta_{t+1}^3$ is close to zero, which corresponds to time steps when the particles fail to cover the support of the emission density. Since the computationally more efficient PaRISian version allows considerably more particles to be used for the given budget, the problem is much smaller for this algorithm; indeed, whereas the learning curves of the Rao-Blackwellised version exhibit a high degree of volatility and several large jumps (the peaks are cut for visibility) with subsequent very long excursions out in the parameter space, the PaRISian version exhibits only one significant jump of moderate size. 

\begin{figure}[htb]
	\begin{minipage}[b]{0.49\linewidth}
		\centering
		\centerline{\includegraphics[width=8cm]{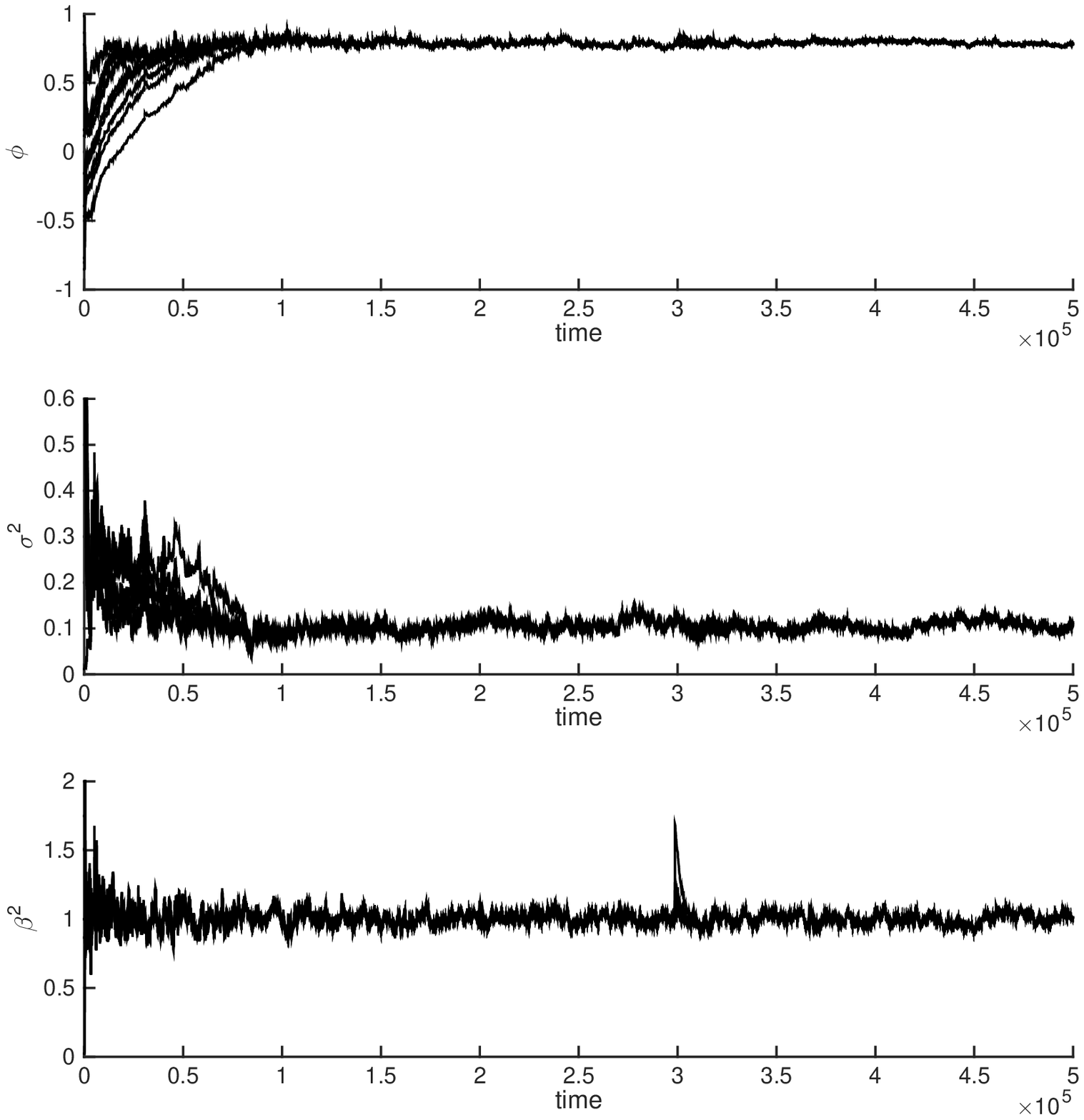}}
		\centerline{(a) PaRIS-based RML}
	\end{minipage}
	\begin{minipage}[b]{0.49\linewidth}
		\centering
		\centerline{\includegraphics[width=8cm]{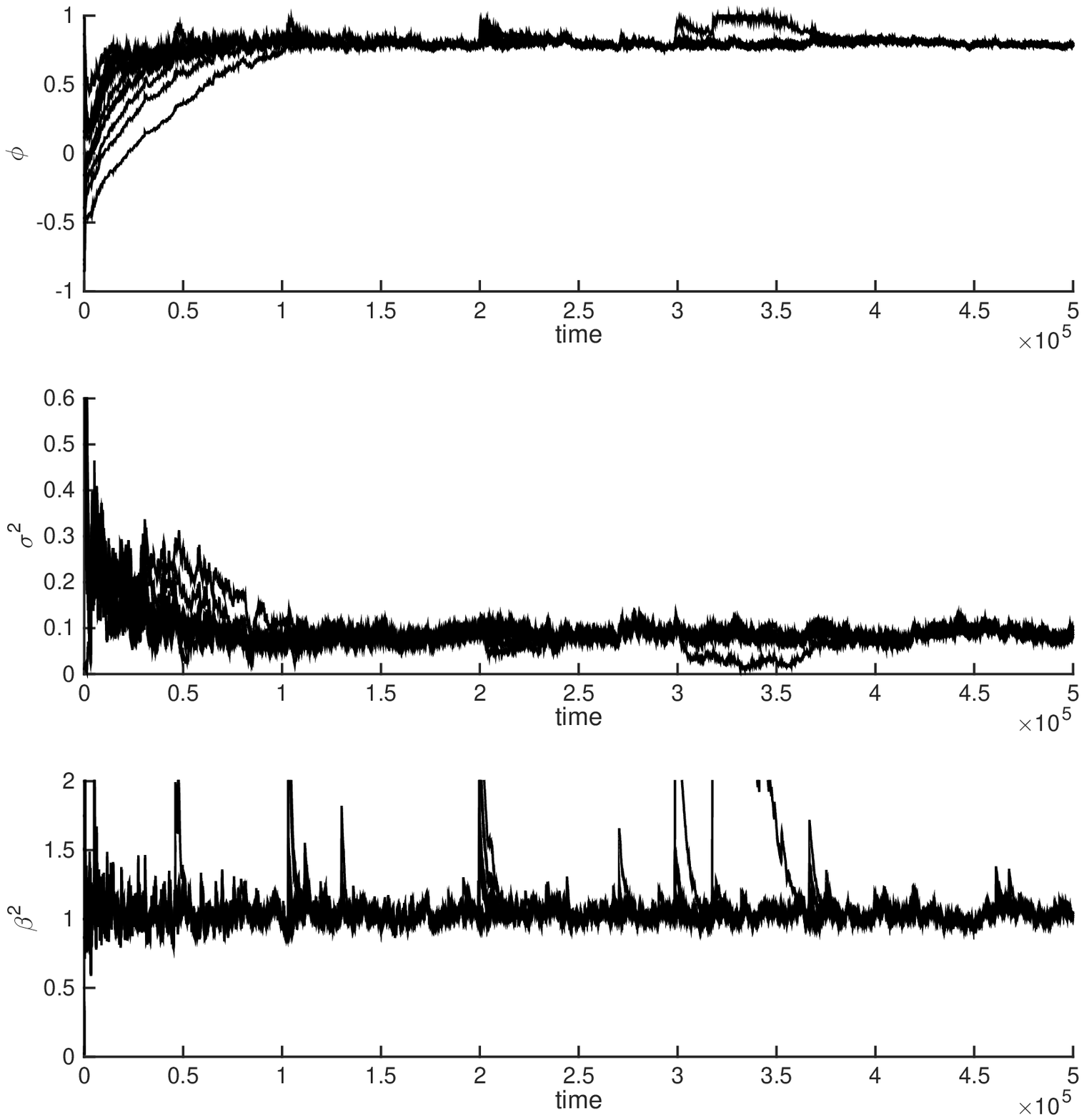}}
		\centerline{(b) Rao-Blackwellised RML}
	\end{minipage}
	\caption{Particle learning trajectories produced by the PaRISian RML estimator (left panel) and the Rao-Blackwellised RML estimator \johan{(right panel)} proposed in \cite{delmoral:doucet:singh:2015} for, from top to bottom, $\phi$, $\sigma^2$, and $\beta^2$. The former and latter algorithms used $\N = 1,\!400$ and $\N = 400$ particles, respectively (leading to comparable CPU times). For each algorithm, $12$ replicates were generated on the same data set with different, randomised initial parameters (being the same for both algorithms). For the Rao-Blackwellised version, the plot of $\beta^2$ does not contain the full trajectories due to very high peaks.}
	\label{fig:trajectories}
\end{figure}

\subsubsection{Simultaneous localisation and mapping (SLAM)}

The simultaneous localisation and mapping (SLAM) problem is fundamental in robotics. In the version considered here we let the state space consist of the coordinates and the bearing of a robot moving in the plane, i.e. $\set{X} = \mathbb{R}^2 \times (-\pi, \pi]$. The prior motion of the robot is modeled as a Markov chain $X_t \eqdef (X^1_t, X^2_t, X^3_t)$, $t \in \nset$, on this space, defined through the state equations
\begin{align}
	X^1_{t+1} & = X^1_{t} + d_{t+1} \cos(X^3_{t}) + \epsilon^1_{t+1}, \\
	X^2_{t+1} & = X^2_{t} + d_{t+1} \sin(X^3_{t}) + \epsilon^2_{t+1}, \\
	X^3_{t+1} & = X^3_{t} + \alpha_{t+1} + \epsilon^3_{t+1},
\end{align}
where $\{\epsilon^i_{t}\}_{t \in \nsetpos}$, $i \in \{1, 2, 3\}$, are independent sequences of mutually independent noise variables with known variances $\sigma^2_{i}$, $i \in \{1, 2, 3\}$. The sequence $\{ ( d_t, \alpha_t ) \}_{t \in \nsetpos}$ provides the commands---in terms of speed and bearing changes---of the robot at each time point.

The robot observes a landmark, defined by its positions in the plane, by measuring the distance and the relative bearing to the same. Assuming $L$ landmarks, the observations at time $t$ are given by $Y_t = \{Y_t^i\}_{i \in O_t}$, where $O_t \subseteq \{1, \ldots, L\}$ is the set of observed landmarks. For each observed landmark,  
$$
	Y^i_{t} = \begin{pmatrix}
		\| (\theta^i_1, \theta^i_2) - (X^1_{t}, X_t^2) \| + \varepsilon_t^{i,1}\\
		\displaystyle \arctan \left( \frac{\theta^i_2 - X^2_{t}}{\theta^i_1 - X^1_{t}} \right) - X^3_{t} + \varepsilon_t^{i,2}
	\end{pmatrix},
$$
where $\| \cdot \|$ denotes the Euclidean norm, $(\theta_1^i, \theta_2^i)$ is the location of landmark $i$, $(\varepsilon_t^{i,1}, \varepsilon_t^{i,2})$ are independent noise variables with known variances $(\varsigma^2_1, \varsigma^2_2)$, the noises of different time steps and different landmarks being independent.

In this setting we wish to estimate the locations of all the landmarks, which implies $2 L$ unknown parameters. Note that the noise parameters in the model are assumed to be calibrated beforehand. Several existing works apply particle methods to the SLAM problem; see, e.g., the \emph{FastSLAM}~\citep{montemerlo:2002} and \emph{FastSLAM 2.0}~\citep{Montemerlo:Thrun:Koller:Wegbreit:2003} algorithms. More recently an online EM version was proposed by \cite{LeCorff:Fort:Moulines:2011} and the \emph{Marginal-SLAM algorithm} \citep{Cantin:Freitas:Castellanos:2007} is based on an updating step similar to~\eqref{eq:par:update}.

Using the model above, data was simulated for $L = 9$ landmarks located according to the green dots in~\autoref{fig:SLAM}. The observations were generated by simulating $100,\!000$ observations of the robot moving around a closed loop with a time resolution of $.1$ s. The noise parameters were set to $\sigma_1 = \sigma_2 = .25$, $\sigma_3 = \frac{3 \pi}{180}$, $\varsigma_1 = .25$, and $\varsigma_2 = \frac{\pi}{180}$. The robot observes landmarks within a $30$ m radius and a $180^{\circ}$ field of vision.

Using this simulated data we performed $20$ independent runs of our PaRISian RML algorithm in order to estimate the landmark positions. For each landmark, the initial estimate of its position was drawn randomly according to a bivariate Gaussian distribution with the true landmark position as mean and the identity matrix as covariance matrix. The algorithm used $(\N, \K) = (500, 2)$ and the particles were initialised in $(0,0)$ with zero bearing, i.e., the same as the robot's starting position. Since the problem is invariant under translations and rotations we fix two landmarks and assume that these are known a priori. The results are displayed in~\autoref{fig:SLAM}, where the estimates cluster closely around the true parameter values. The figure also shows the estimated filtered trajectory of the robot during its first $3$ laps. In order to illustrate the convergence of the landmark position estimates, we consider the evolution of their MSE over time. The results are reported in~\autoref{fig:SLAM:mse}, where the average MSE of $4$ landmarks is computed based on the $20$ independent runs of the algorithm. As clear from the picture, the processing of the full observation record yields MSE values that fluctuate stably around zero. 

{Needless to say, even though this prototypical solution to the SLAM seems promising, a lot of additional work is needed for obtaining an algorithm running in real applications. Moreover, the comprehensive task of benchmarking our approach against full palette of existing algorithms treating this fundamental problem falls outside the scope of the present paper.}

\begin{figure}[htb]
	\includegraphics[width = 12 cm]{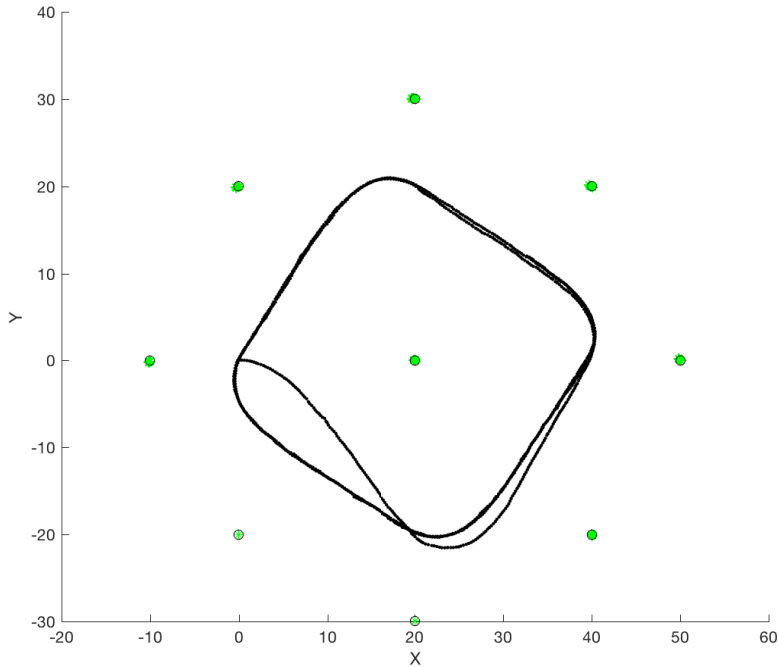}
	\caption{Resulting estimates of the $9$ landmarks for the SLAM problem using the PaRISian RML algorithm. The circles are the true positions of the landmarks and green stars are the resulting estimates. The black dots denote the particle estimates of the robot's positions during the first three laps. Rerunning the algorithm yields  similar trajectories.}
	\label{fig:SLAM}
\end{figure}

\begin{figure}[htb]
	\includegraphics[width = 12 cm]{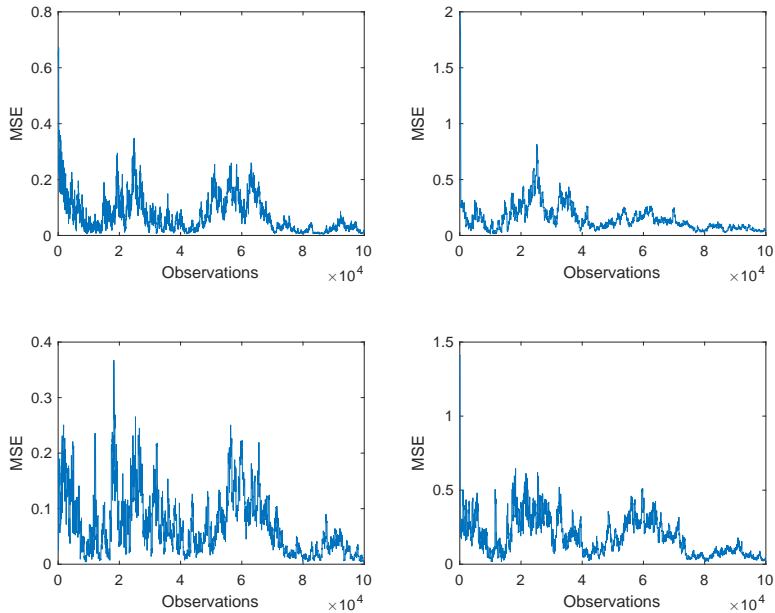}
	\caption{Resulting mean square error estimate for 4 of the landmarks based on the 20 independent runs of the PaRISian RML algorithm for the SLAM problem.}
	\label{fig:SLAM:mse}
\end{figure}



\section{Conclusion} 
\label{sec:discussion}

We have presented a novel algorithm for efficient particle-based estimation of tangent filters in general SSMs. The algorithm involves the PaRIS algorithm proposed by \cite{olsson:westerborn:2014b} as a key ingredient. The estimator is furnished with several convergence results, the main result being a CLT at the rate $\sqrt{\N}$. The convergence analysis is driven by results obtained by \cite{olsson:westerborn:2014b, delmoral:doucet:singh:2015}, and our time uniform bound on the asymptotic variance extends the existing results in \cite{delmoral:doucet:singh:2015}. Importantly, {under weak assumptions,} the computational complexity of our algorithm grows only linearly with the number of particles, whereas the complexity of existing methods such as the estimator of \cite{delmoral:doucet:singh:2015}, which may be viewed as a Rao-Blackwellisation of our estimator, is typically quadratic in the number of particles. Thus, we may expect the computational benefit of the ``de-Rao-Blackwellisation'' that characterises PaRIS to grossly exceed the price of the additional variance introduced by the PaRISian backward sampling procedure. 

In the second part of the paper we cast our PaRISian tangent filter estimator into the framework of RML, yielding a computationally efficient and easily implemented online parameter estimation procedure. As clear from the simulations, the fact that our online estimator allows, compared to existing methods, more particles to be used for a given computational budget is of importance for the stability of the online estimates. 

{Appealingly, the asymptotic variance of our estimator can be bounded uniformly in time; in other words, the estimator stays stochastically stable in the long term. Needless to say, the strong mixing assumptions driving the stability analysis are restrictive, and to relax these assumptions is a natural direction of research. \cite{jasra:2015} provides an $O(t)$ bound on the asymptotic variance of the \emph{forward-filtering backward-smoothing} (FFBSm) \emph{algorithm} \citep{doucet:godsill:andrieu:2000}, which is equivalent with the Rao-Blackwellised PaRIS in the case of additive state functionals (the online formulation was however found by \citet{delmoral:doucet:singh:2010}), under assumptions that point to applications on non-compact state spaces, and the same approach may be applicable to our PaRISian tangent filter estimator.}

{Finally, it still remains to prove, under general, verifiable assumptions, the convergence of the parameter estimates produced by our PaRISian RML algorithm. Such an analysis is expected to be technically very involved, and is hence beyond the scope of the present paper.}  In~\cite{legland:mevel:1997}, in which the convergence of exact RML in HMMs is established for the case of finite state space, the proof consists in showing that the extended Markov chain comprising the state and observation processes as well as the prediction and tangent filter flows is ergodic and applying standard stochastic approximation results~\cite{Deylon:1996}. A route to the convergence of our proposed PaRISian RML could be to include also the particle cloud and the particle-based statistics $\{ \tau_t^i \}_{i = 1}^\N$ into the Markov chain and study the ergodicity of the same. {An alternative and more direct approach would be to cast the problem into the framework of \emph{biased stochastic gradient search}, which was analysed recently by \cite{tadic:2017}.}

\bibliographystyle{abbrvnat}
\bibliography{biblio}

\appendix
\section{Proofs}
\label{sec:proofs}

Define for all $t \in \nset$ and $\theta \in \Theta$, 
$$
	\lk{t;\parvec} : \set{X} \times \alg{X} \ni (x, A) \mapsto  \md{t;\parvec}(x) \hk_{\parvec}(x, A). 
$$
(Note that our definition of $\lk{t}$ differs from that used by~\cite{olsson:westerborn:2014b}, in which the order of $\md{t;\parvec}$ and $\hk_{\parvec}$ is swapped.) With this notation, by the filtering recursion \eqref{eq:pred:to:filt}--\eqref{eq:filt:to:pred},   
\begin{equation}
	\pred{t+1;\parvec} = \frac{\pred{t;\parvec} \lk{t;\parvec}}{\pred{t;\parvec} \lk{t;\parvec}\1{\set{X}}}, \label{eq:pred:rec}
\end{equation}
with, as previously, $\post{0 \mid -1} \eqdef \Xinit$. This condensed form of the filtering recursion will be used in \autoref{sec:proof:CLT}.  

In the coming analysis the following decomposition will be instrumental. For all $t \in \nset$, 
\begin{multline} \label{eq:key:decomposition}
	\filtderiv[part]{t;\parvec} \testf[t] - \filtderiv{t;\parvec} \testf[t] 
	=  \frac{1}{\N} \sum_{i = 1}^{\N} \{ (\tstat[i]{t}  - \pred{t;\parvec} \tstat{t;\parvec} \af{t;\parvec} ) (\testf[t](\epart{t}{i})  - \pred{t;\parvec} \testf[t]) \} 
	- \pred{t;\parvec} \{ (\tstat{t;\parvec} \af{t;\parvec} - \pred{t;\parvec} \tstat{t;\parvec} \af{t;\parvec}) (\testf[t] - \pred{t;\parvec} \testf[t]) \} \\
	- \left( \frac{1}{\N} \sum_{i = 1}^{\N} \tstat[i]{t} - \pred{t;\parvec} \tstat{t;\parvec} \af{t;\parvec} \right) \left( \frac{1}{\N} \sum_{i = 1}^{\N} \testf[t](\epart{t}{i}) - \pred{t;\parvec} \testf[t] \right). 
\end{multline}

\subsection{Proof of~\autoref{thm:hoeffding}}
We apply the decomposition \eqref{eq:key:decomposition}. Note that  
\begin{multline} \label{eq:second:key:identity}
\frac{1}{\N} \sum_{i = 1}^{\N} \{ (\tstat[i]{t}  - \pred{t;\parvec} \tstat{t;\parvec} \af{t;\parvec} ) (\testf[t](\epart{t}{i})  - \pred{t;\parvec} \testf[t]) \} - \pred{t;\parvec} \{ (\tstat{t;\parvec} \af{t;\parvec} - \pred{t;\parvec} \tstat{t;\parvec} \af{t;\parvec}) (\testf[t] - \pred{t;\parvec} \testf[t]) \} \\
=  \frac{1}{\N} \sum_{i = 1}^{\N} \wgt{t}{i} \{ \tstat[i]{t} \testfp[t;\parvec](\epart{t}{i}) + \testfh[t;\parvec](\epart{t}{i}) \} - \pred{t;\parvec} \{ \md{t;\parvec} (\tstat{t;\parvec} \af{t;\parvec} \testfp[t;\parvec] + \testfh[t;\parvec]) \},
\end{multline}
where
\begin{equation} \label{eq:def:mod:test}
\begin{split}
	\testfp[t;\parvec](x) &\eqdef \frac{\testf[t](x) - \pred{t;\parvec} \testf[t]}{\md{t;\parvec}(x)},\\
	\testfh[t;\parvec](x) &\eqdef - \frac{\pred{t;\parvec} \tstat{t;\parvec} \af{t;\parvec}\{ \testf[t](x) - \pred{t;\parvec} \testf[t] \}}{\md{t;\parvec}(x)},
\end{split} \qquad x \in \set{X}. 
\end{equation}
Now, since $\testfp[t] \md{t;\parvec}$ and $\testfh[t] \md{t;\parvec}$ both belong to $\bmf{\alg{X}}$, \autoref{lemma:hoeffding} provides constants $c_t > 0$ and $\tilde{c}_t > 0$ such that for all $\varepsilon > 0$, 
\begin{equation} \label{eq:hoeffding-1}
	\prob \left( \left| \frac{1}{\N} \sum_{i = 1}^{\N} \wgt{t}{i} \{ \tstat[i]{t} \testfp[t;\parvec](\epart{t}{i}) + \testfh[t;\parvec](\epart{t}{i}) \} -  \pred{t;\parvec} \{ \md{t;\parvec} (\tstat{t;\parvec} \af{t;\parvec} \testfp[t;\parvec] + \testfh[t;\parvec]) \} \right| \geq \varepsilon \right) \leq c_t \exp{(- \tilde{c}_t \N \varepsilon^2)}.
\end{equation}
To deal with the second part of the decomposition \eqref{eq:key:decomposition}, we use the same technique. First, by applying \autoref{lemma:hoeffding} with $f \equiv 1/\md{t;\parvec}$ and $\tilde{f} \equiv 0$, we obtain constants $a_t > 0$ and $\tilde{a}_t > 0$ such that for all $\varepsilon > 0$, 
\begin{equation} \label{eq:hoeffding-2}
\prob\left( \left| \frac{1}{\N} \sum_{i = 1}^{\N} \tstat[i]{t} - \pred{t;\parvec} \tstat{t;\parvec} \af{t;\parvec} \right| \geq \varepsilon \right) \leq a_t \exp{{(- \tilde{a}_t \N \varepsilon^2)}}.
\end{equation}
Similarly, using \autoref{lemma:hoeffding} with $\testf \equiv 0$ and $\tilde{f} \equiv f_t / \md{t;\parvec}$ provides constants $b_t > 0$ and $\tilde{b}_t > 0$ such that for all $\varepsilon > 0$, 
\begin{equation} \label{eq:hoeffding-3}
\prob \left( \left| \frac{1}{\N} \sum_{i=1}^{\N} \testf[t](\epart{t}{i}) - \pred{t;\parvec}\testf[t] \right| \geq \varepsilon \right) \leq b_t \exp{(- \tilde{b}_t \N \varepsilon^2)}. 
\end{equation}
Combining \eqref{eq:hoeffding-1}, \eqref{eq:hoeffding-2}, and \eqref{eq:hoeffding-3} yields, for all $\varepsilon > 0$,  
$$
\prob \left( \left| \filtderiv[part]{t;\parvec} \testf[t] - \filtderiv{t;\parvec} \testf[t] \right| \geq \varepsilon \right) \leq a_t \exp{(- \tilde{a}_t \N \varepsilon / 2 )} + b_t \exp{(- \tilde{b}_t \N \varepsilon / 2)} + c_t \exp \{- \tilde{c}_t \N (\varepsilon / 2)^2 \}, 
$$
from which the statement of the theorem follows. \qed

The following result is obtained by inspection of the proof of \citet[Theorem~1(i)]{olsson:westerborn:2014b}.  
\begin{proposition} \label{lemma:hoeffding}
Let \autoref{ass:bounded:q:g} hold. Then for all $t \in \nset$, all $\parvec \in \Theta$, all additive state functionals $\af{t} \in \bmf{\alg{X}^{\varotimes (t + 1)}}$, all measurable functions $f_t$Êand $\tilde{f}_t$ such that $f_t \md{t;\parvec} \in \bmf{\alg{X}}$ and $\tilde{f}_t \md{t;\parvec} \in \bmf{\alg{X}}$, and all $\K \in \nset$, there exist constants $c_t > 0$ and $\tilde{c}_t > 0$ (possibly depending on $\parvec$, $\af{t}$ $f_t$, $\tilde{f}_t$, and $\K$) such that for all $\varepsilon > 0$, 
$$
\prob \left( \left| \frac{1}{\N} \sum_{i = 1}^{\N} \wgt{t}{i} \{ \tstat[i]{t} \testf[t](\epart{t}{i}) + \tilde{f}_t(\epart{t}{i}) \} -  \pred{t; \parvec} \{ \md{t;\parvec} (\tstat{t;\parvec} \af{t} \testf[t] + \tilde{f}_t) \} \right| \geq \varepsilon \right) \\
		 \leq c_t \exp{(- \tilde{c}_t \N \varepsilon^2)},
$$
where $\{ (\epart{t}{i}, \wgt{t}{i}, \tstat[i]{t}) \}_{i = 1}^\N$ are produced using the PaRIS algorithm. 
\end{proposition}

\subsection{Proof of \autoref{corollary:as:conv}}

The $\prob$-a.s. convergence of $\filtderiv[part]{t; \parvec} \testf[t]$ to $ \filtderiv{t; \parvec} \testf[t]$ is implied straightforwardly by the exponential convergence rate in \autoref{thm:hoeffding}. Indeed, note that 
$$
\prob \left( \lim_{\N \rightarrow \infty} \filtderiv[part]{t; \parvec} \testf[t] \neq \filtderiv{t; \parvec} \testf[t] \right) 
= \lim_{k \to \infty} \lim_{n \to \infty} \prob \left( \bigcup_{n \leq \N} \left\{ \left| \filtderiv[part]{t; \parvec} \testf[t] - \filtderiv{t; \parvec} \testf[t] \right| \geq \frac{1}{k} \right\} \right);  
$$
now, by \autoref{thm:hoeffding}, 
$$
\prob \left( \bigcup_{n \leq \N} \left\{ \left| \filtderiv[part]{t; \parvec} \testf[t] - \filtderiv{t; \parvec} \testf[t] \right| \geq \frac{1}{k} \right\} \right) \leq c_t \exp \left( - \tilde{c}_t n k^{-2} \right) \sum_{\N = 0}^\infty \exp \left( - \tilde{c}_t \N k^{-2} \right), 
$$
where the right hand side tends to zero when $n$ tends to infinity. This completes the proof.  \qed

\subsection{Proof of~\autoref{thm:CLT}} \label{sec:proof:CLT}
 
By combining \eqref{eq:key:decomposition} and \eqref{eq:second:key:identity},  
\begin{multline} 
	\sqrt{\N} \left( \filtderiv[part]{t;\parvec} \testf[t] - \filtderiv{t;\parvec} \testf[t] \right) =  \sqrt{\N} \left( \frac{1}{\N} \sum_{i = 1}^{\N} \{ \tstat[i]{t} \testfp[t;\parvec](\epart{t}{i}) + \testfh[t;\parvec](\epart{t}{i}) \} - \pred{t;\parvec}(\tstat{t;\parvec} \af{t;\parvec} \testfp[t;\parvec] + \testfh[t;\parvec]) \right) \\ 
	- \sqrt{\N} \left( \frac{1}{\N} \sum_{i = 1}^{\N} \testf[t](\epart{t}{i}) - \pred{t;\parvec} \testf[t] \right) 
	\left( \frac{1}{\N} \sum_{i = 1}^{\N} \tstat[i]{t} - \pred{t;\parvec} \tstat{t;\parvec} \af{t;\parvec} \right),
\end{multline}
where in this case
\[
	\begin{split}
		\testfp[t;\parvec](x) &\eqdef \testf[t](x) - \pred{t;\parvec} \testf[t],\\
		\testfh[t;\parvec](x) &\eqdef - \pred{t;\parvec} \tstat{t;\parvec} \af{t;\parvec} \{ \testf[t](x) - \pred{t;\parvec} \testf[t] \},
	\end{split} \qquad x \in \set{X}. 
	\]
are defined in \eqref{eq:def:mod:test}. By \autoref{lemma:CLT}, since $\testfp[t;\parvec] \in \bmf{\set{X}}$ and $\testfh[t;\parvec] \in \bmf{\set{X}}$, 
$$
\sqrt{\N} \left( \frac{1}{\N} \sum_{i = 1}^{\N} \{ \tstat[i]{t} \testfp[t;\parvec](\epart{t}{i}) + \testfh[t;\parvec](\epart{t}{i}) \} - \pred{t;\parvec}(\tstat{t;\parvec} \af{t;\parvec} \testfp[t;\parvec] + \testfh[t;\parvec]) \right) \convd \sigma_{t;\parvec} \langle \testfp[t;\parvec], \testfh[t;\parvec] \rangle(\af{t}) Z, 
$$
where $Z$ is standard normally distributed and 
$$
 \sigma_{t;\parvec} \langle \testfp[t;\parvec], \testfh[t;\parvec] \rangle(\af{t;\parvec}) = \sigma_{t;\parvec}(\af{t;\parvec}), 
$$
with $\sigma_{t;\parvec}(\af{t;\parvec})$ being defined in \eqref{eq:asvar:short}. Now, \autoref{lemma:hoeffding} and \autoref{lemma:CLT} yield
$$
\frac{1}{\N} \sum_{i = 1}^{\N} \tstat[i]{t} \convp \pred{t;\parvec} \tstat{t;\parvec} \af{t;\parvec}, 
$$
and
$$
\sqrt{\N} \left( \frac{1}{\N} \sum_{i = 1}^{\N} \testf[t](\epart{t}{i}) - \pred{t;\parvec} \testf[t] \right) \convd \sigma_{t;\parvec}^2 \langle 0, \testf[t] \rangle(\af{t;\parvec}) Z
$$
(with $0$ denoting the zero function), respectively, implying, by Slutsky's theorem, 
$$
\sqrt{\N} \left( \frac{1}{\N} \sum_{i = 1}^{\N} \testf[t](\epart{t}{i}) - \pred{t;\parvec} \testf[t] \right) 
	\left( \frac{1}{\N} \sum_{i = 1}^{\N} \tstat[i]{t} - \pred{t;\parvec} \tstat{t;\parvec} \af{t;\parvec} \right) \convp 0. 
$$
Finally, we complete the proof by noting that the term $\tilde{\sigma}_{t;\parvec}^2(\testf[t])$ in \eqref{eq:asvar:short} coincides with the asymptotic variance provided by~\citet[Theorem~3.2]{delmoral:doucet:singh:2015}. \qed

\begin{proposition} \label{lemma:CLT}
\autoref{ass:bounded:q:g} hold. Then for all $t \in \nset$, all $\parvec \in \Theta$, all additive state functionals $\af{t} \in \bmf{\alg{X}^{\varotimes (t + 1)}}$, all measurable functions $f_t \in \bmf{\alg{X}}$ and $\tilde{f}_t \in \bmf{\alg{X}}$, and all $\K \in \nset$,  as $\N \to \infty$, 
	\begin{equation}
		\sqrt{N} \left( \frac{1}{\N} \sum_{i=1}^{\N} \{ \tstat[i]{t} \testf[t](\epart{t}{i}) + \tilde{f}_t(\epart{t}{i}) \} - \pred{t; \parvec} (\tstat{t;\parvec} \af{t} \testf[t] + \tilde{f}_t) \right) \convd{} \sigma_{t;\parvec} \langle f_t, \tilde{f}_t \rangle(\af{t}) Z,
	\end{equation}
	where $Z$ is a standard Gaussian random variable and 
	\begin{equation} \label{eq:variance:expression:lemma}
		\sigma_{t;\parvec}^2 \langle f_t, \tilde{f}_t \rangle(\af{t}) \eqdef 
		 \tilde{\sigma}_{t;\parvec}^2 \langle f_t, \tilde{f}_t \rangle(\af{t}) 
		 + \sum_{s = 0}^{t-1} \sum_{\ell = 0}^{s} \K^{\ell - (s+1)} \varsigma_{s, \ell, t; \parvec} \langle f_t \rangle(\af{t}), 
	\end{equation}
	with
	$$
	\sigma_{t;\parvec}^2 \langle f_t, \tilde{f}_t \rangle(\af{t}) \eqdef \sum_{s=0}^{t-1} \frac{\pred{s+1;\parvec}\BFcent[2]{s+1}{t;\parvec}(\af{t} f_t + \tilde{f}_t)}{(\pred{s+1;\parvec} \lk{s+1;\parvec} \cdots \lk{t-1;\parvec} \1{\set{X}})^2}
	$$
	and 
	$$
	\varsigma_{s, \ell, t; \parvec} \langle f_t \rangle(\af{t}) \eqdef \frac{\pred{\ell+1;\parvec} \{ \bk{\post{\ell;\parvec}}(\tstat{\ell;\parvec} \af{\ell} + \addf{\ell} - \tstat{\ell+1;\parvec} \af{\ell + 1})^2 \lk{\ell + 1;\parvec} \cdots \lk{s;\parvec}( \lk{s+1;\parvec} \cdots \lk{t-1;\parvec} f_t)^2 \} }{(\pred{\ell+1;\parvec} \lk{\ell + 1;\parvec} \cdots \lk{s;\parvec} \1{\set{X}})(\pred{s+1;\parvec} \lk{s+1;\parvec} \cdots \lk{t-1;\parvec} \1{\set{X}})^2}. 
	$$
\end{proposition}

\begin{proof}[of \autoref{lemma:CLT}]
Assume first that $\pred{t; \parvec} (\tstat{t;\parvec} \af{t} \testf[t] + \tilde{f}_t) = 0$. Then, by \autoref{lemma:lemma:CLT} and Slutsky's theorem, as $\wgtsum{t} / \N  \convp \pred{t;\parvec} \md{t;\parvec}$ by \autoref{lemma:hoeffding}, 
\begin{multline}
\sqrt{\N} \left( \frac{1}{\N} \sum_{i = 1}^{\N} \{ \tstat[i]{t} \testf[t](\epart{t}{i}) + \tilde{f}_t(\epart{t}{i}) \} \right) \\Ê= \sqrt{\N} \left( \sum_{i = 1}^{\N} \frac{\wgt{t}{i}}{\wgtsum{t}} \{ \tstat[i]{t} \testfp[t;\parvec](\epart{t}{i}) + \testfh[t;\parvec](\epart{t}{i}) \} \right) \frac{1}{\N} \wgtsum{t} \convd 
\pred{t;\parvec} \md{t;\parvec} \times \Gamma_{t;\parvec} \langle \testfp[t;\parvec], \testfh[t;\parvec] \rangle(\af{t}) Z, 
\end{multline}
where again $Z$ has standard Gaussian distribution, $\Gamma_{t;\parvec} \langle \testfp[t;\parvec], \testfh[t;\parvec] \rangle(\af{t})$ is given in \autoref{lemma:lemma:CLT}, and we have set $\testfp[t;\parvec] \eqdef f_t / \md{t;\parvec}$ and $\testfh[t;\parvec] \eqdef \tilde{f}_t / \md{t;\parvec}$ and used, first, that $\md{t:\parvec} \testfp[t;\parvec] = f_t \in \bmf{\alg{X}}$ and $\md{t:\parvec} \testfh[t;\parvec] = \tilde{f}_t \in \bmf{\alg{X}}$ and, second, that 
$$
\post{t;\parvec}(\tstat{t;\parvec} \af{t;\parvec} \testfp[t;\parvec] + \testfh[t;\parvec]) = \frac{\pred{t;\parvec} \{ \md{t;\parvec} (\tstat{t;\parvec} \af{t;\parvec} \testfp[t;\parvec] + \testfh[t;\parvec])\}}{\pred{t;\parvec} \md{t;\parvec}} = \frac{\pred{t; \parvec} (\tstat{t;\parvec} \af{t} \testf[t] + \tilde{f}_t)}{\pred{t;\parvec} \md{t;\parvec}} = 0. 
$$ 
Now, by iterating \eqref{eq:pred:rec} we conclude that for all $(s, t) \in \nset^2$,  
\begin{multline}
\pred{t;\parvec} = \frac{\pred{t - 1;\parvec} \lk{t - 1;\parvec}}{\pred{t - 1;\parvec} \lk{t - 1;\parvec} \1{\set{X}}} = \frac{\pred{t - 2;\parvec} \lk{t - 2;\parvec} \lk{t - 1;\parvec}}{(\pred{t - 2;\parvec} \lk{t - 2;\parvec} \1{\set{X}}) (\pred{t - 1;\parvec} \lk{t - 1;\parvec} \1{\set{X}})} = \frac{\pred{t - 2;\parvec} \lk{t - 2;\parvec} \lk{t - 1;\parvec}}{\pred{t - 2;\parvec} \lk{t - 2;\parvec} \lk{t - 1;\parvec} \1{\set{X}}} \\ 
= \frac{\pred{s + 1;\parvec} \lk{s + 1;\parvec} \ldots \lk{t - 1;\parvec}}{\pred{s + 1;\parvec} \lk{s + 1;\parvec} \ldots \lk{t - 1;\parvec} \1{\set{X}}} 
\end{multline}
and, consequently, 
\begin{equation} \label{eq:one:step:likelihood:alt}
\pred{t;\parvec} \md{t;\parvec} = \frac{\pred{s + 1;\parvec} \lk{s + 1;\parvec} \ldots \lk{t - 1;\parvec} \md{t;\parvec}}{\pred{s + 1;\parvec} \lk{s + 1;\parvec} \ldots \lk{t - 1;\parvec} \1{\set{X}}}. 
\end{equation}
Finally, by \eqref{eq:one:step:likelihood:alt} it holds that 
$$
\pred{t;\parvec} \md{t;\parvec} \times \Gamma_{t;\parvec} \langle \testfp[t;\parvec], \testfh[t;\parvec] \rangle(\af{t}) = \sigma_{t;\parvec} \langle f_t, \tilde{f}_t \rangle(\af{t}), 
$$
where $\Gamma^2_{t;\parvec} \langle f_t, \tilde{f}_t \rangle(\af{t})$ is defined in~\eqref{eq:variance:expression:lemma:gamma} and $\sigma^2_{t;\parvec} \langle f_t, \tilde{f}_t \rangle(\af{t})$ is defined in~\eqref{eq:variance:expression:lemma}. Finally, in the general case, the previous holds true when $\tilde{f}_t$ is replaced by $\tilde{f}_t - \pred{t; \parvec} (\tstat{t;\parvec} \af{t} \testf[t] + \tilde{f}_t)$, which completes the proof. \qed
\end{proof}

The following lemma is obtained by inspection of the proof of \citet[Theorem~3]{olsson:westerborn:2014b}. 

\begin{lemma} \label{lemma:lemma:CLT}
\autoref{ass:bounded:q:g} hold. Then for all $t \in \nset$, all $\parvec \in \Theta$, all additive state functionals $\af{t} \in \bmf{\alg{X}^{\varotimes (t + 1)}}$, all measurable functions $f_t$Êand $\tilde{f}_t$ such that $f_t \md{t;\parvec} \in \bmf{\alg{X}}$ and $\tilde{f}_t \md{t;\parvec} \in \bmf{\alg{X}}$, and all $\K \in \nset$,  as $\N \to \infty$, 
	\begin{equation}
		\sqrt{N} \left( \sum_{i=1}^{\N} \frac{\wgt{t}{i}}{\wgtsum{t}} \{ \tstat[i]{t} f_t(\epart{t}{i}) + \tilde{f}_t(\epart{t}{i}) \} - \post{t; \parvec} (\tstat{t;\parvec} \af{t} \testf[t] + \tilde{f}_t) \right) \convd{} \Gamma_{t;\parvec} \langle f_t, \tilde{f}_t \rangle(\af{t}) Z,
	\end{equation}
	where $Z$ is a standard normal distribution and 
	\begin{equation} \label{eq:variance:expression:lemma:gamma}
		\Gamma_{t;\parvec}^2 \langle f_t, \tilde{f}_t \rangle(\af{t}) \eqdef 
		 \tilde{\Gamma}_{t;\parvec}^2 \langle f_t, \tilde{f}_t \rangle(\af{t}) 
		 + \sum_{s = 0}^{t-1} \sum_{\ell = 0}^{s} \K^{\ell - (s+1)} \gamma_{s, \ell, t; \parvec} \langle f_t \rangle(\af{t}), 
	\end{equation}
	with
	$$
	\tilde{\Gamma}_{t;\parvec}^2 \langle f_t, \tilde{f}_t \rangle(\af{t}) \eqdef \sum_{s=0}^{t-1} \frac{\pred{s+1;\parvec}\BFcent[2]{s+1}{t;\parvec}\{\md{t;\parvec} (\af{t} f_t + \tilde{f}_t)\}}{(\pred{s+1;\parvec} \lk{s+1;\parvec} \cdots \lk{t-1;\parvec} \md{t;\parvec})^2}
	$$
	and
	$$
	\gamma_{s, \ell, t; \parvec} \langle f_t \rangle(\af{t}) \eqdef \frac{\pred{\ell+1;\parvec} \{ \bk{\post{\ell;\parvec}}(\tstat{\ell;\parvec} \af{\ell} + \addf{\ell} - \tstat{\ell+1;\parvec} \af{\ell + 1})^2 \lk{\ell + 1;\parvec} \cdots \lk{s;\parvec}( \lk{s+1;\parvec} \cdots \lk{t-1;\parvec} \md{t;\parvec} f_t)^2 \} }{(\pred{\ell+1;\parvec} \lk{\ell + 1;\parvec} \cdots \lk{s;\parvec} \1{\set{X}})(\pred{s+1;\parvec} \lk{s+1;\parvec} \cdots \lk{t-1;\parvec} \md{t;\parvec})^2}. 
	$$
\end{lemma}

\subsection{Proof of~\autoref{thm:bdd:var}}
As noted above, the first term of the asymptotic variance $\asvard[2]{t}{\testf[t]}$ coincides with the asymptotic variance $\tilde{\sigma}_{t; \theta}^2(\testf[t])$ obtained by~\citet[Theorem 3.2]{delmoral:doucet:singh:2015}. {The same work provides a constant $c \in \rset_+$ such that $\tilde{\sigma}_{t; \theta}^2(\testf[t]) \leq c \supn{\testf[t]}$, and} we may hence focus on bounding second term of the asymptotic variance.  

For this purpose, note that for all $s \leq t - 1$ and $x_{s + 1} \in \set{X}$,  
	\begin{multline}
		\lk{s+1;\parvec} \cdots \lk{t-1;\parvec}(\testf[t] - \pred{t;\parvec}\testf[t])(x_{s+1}) = \md{s+1;\parvec}(x_{s+1}) \hk_{\parvec} \lk{s+2;\parvec} \cdots \lk{t-2;\parvec}\{\md{t-1;\parvec}(\hk_{\parvec} \testf[t] - \post{t-1;\parvec}\hk_{\parvec} \testf[t])\}(x_{s+1}). \label{eq:rewrite}
	\end{multline}
By applying the forgetting of the filter, or, more particularly, \citet[Lemma 10]{douc:garivier:moulines:olsson:2010}, to \eqref{eq:rewrite} we obtain 
	\begin{align}
		\supn{\lk{s+1;\parvec} \cdots \lk{t-1;\parvec}(\testf[t] - \pred{t;\parvec}\testf[t])} \leq 2 \mdup \supn{\hk_{\parvec} \lk{s+2;\parvec} \cdots \lk{t-1;\parvec}\1{\set{X}}}  \| \testf[t] \|_\infty^2 \mr^{t - s - 2}.  
	\end{align}
	{Note that in the previous bound, the exponential contraction follows from the fact that the objective function $f_t$ is centered around its predicted mean. The latter is a consequence of the fact that the tangent filter is, as a covariance, centered itself (recall the identities \eqref{eq:tangent:identity:first:step} and \eqref{eq:tangent:identity} and the decomposition \eqref{eq:key:decomposition}).}
	In addition, from the proof of~\citet[Theorem 8]{olsson:westerborn:2014b} we extract, {using~\autoref{ass:bdd},}
	$$
		\supn{\tstat{\ell;\parvec} \af{\ell;\parvec} + \addf{\ell;\parvec} - \tstat{\ell+1;\parvec}\af{\ell + 1; \parvec}} \leq \hbd \left( 4 \mdup \frac{\mr^{2}}{1 - \mr} + 1 \right), 
	$$
	and under~\autoref{ass:mixing}, for all $x \in \set{X}$,
	$$
	\rho^{-1} \mdlow \refM \lk{s+2;\parvec} \cdots \lk{t-1;\parvec}\1{\set{X}} \leq \lk{s+1;\parvec} \cdots \lk{t-1;\parvec} \1{\set{X}}(x) \leq \rho \mdup \refM \lk{s+2;\parvec} \cdots \lk{t-1;\parvec}\1{\set{X}}.
	$$
Combining the previous bounds gives 
	\begin{multline}
		\asvardPaRIS{s}{\ell}{t}{\testf[t]} = \frac{\pred{\ell+1;\parvec} \{ \bk{\post{\ell;\parvec}}(\tstat{\ell;\parvec} \af{\ell;\parvec} + \addf{\ell;\parvec} - \tstat{\ell+1;\parvec} \af{\ell + 1;\parvec})^2 \lk{\ell + 1;\parvec} \cdots \lk{s;\parvec}( \lk{s+1;\parvec} \cdots \lk{t-1;\parvec} \{\testf[t] - \pred{t;\parvec} \testf[t]\})^2\} }{(\pred{\ell+1;\parvec} \lk{\ell + 1;\parvec} \cdots \lk{s;\parvec} \1{\set{X}})(\pred{s+1;\parvec} \lk{s+1;\parvec} \cdots \lk{t-1;\parvec} \1{\set{X}})^2} \\
		\leq \tilde{c} \|Êf_t \|_\infty^2 \mr^{2(t - s - 2)},
	\end{multline}
where 
$$
\tilde{c} \eqdef 4 \hbd^2 \left( 4 \mdup \frac{\mr^2}{1 - \mr} + 1\right)^2 \left( \frac{\mdup}{\mdlow(1 - \mr)} \right)^3. 
$$
Finally, summing up yields 
	\begin{multline}
		\sum_{s=0}^{t-1}\sum_{\ell = 0}^s \asvardPaRIS{s}{\ell}{t}{\testf[t]}  \leq \tilde{c} \|Êf_t \|_\infty^2 \sum_{s=0}^{t-1}\sum_{\ell = 0}^s \K^{\ell - s- 1}\mr^{2(t-s-2)} = \tilde{c} \|Êf_t \|_\infty^2 \frac{1}{\K - 1} \sum_{s=0}^{t-1} \mr^{2(t-s-2)}(1 - \K^{-s-1}) \\
		\leq \tilde{c} \|Êf_t \|_\infty^2 \frac{1}{\K - 1} \sum_{s=0}^{t-1} \mr^{2(t-s-2)} \leq \tilde{c} \|Êf_t \|_\infty^2 \frac{1}{(\K - 1)(1 - \mr^2)},
	\end{multline}
	which completes the proof. \qed


\section{Kernels}
\label{sec:kernel}

Given two measurable spaces $\measSpace{X}$ and $\measSpace{Y}$, an unnormalised transition kernel $\kernel{K}$ between these spaces induces two integral operators, one acting on functions and the other on measures. Specifically, we define the function
\begin{align}
	\kernel{K} h : \set{X} \ni x \mapsto \int h(x,y) \, \kernel{K}(x, \rmd y) \quad (h \in \bmf{\alg{X} \varotimes \alg{Y}})
\end{align}
and the measure
\begin{align}
	\nu \kernel{K} : \alg{Y} \ni \set{A} \mapsto \int \kernel{K}(x, \set{A}) \, \nu(\rmd x) \quad (\nu \in \meas{\alg{X}})
\end{align}
whenever these quantities are well-defined. Moreover, let $\kernel{L}$ be another unnormalised transition kernel from $\measSpace{Y}$ to the measurable space $\measSpace{Z}$; then two different \emph{products} of $\kernel{K}$ and $\kernel{L}$ can be defined, namely
\begin{align}
	\kernel{K} \kernel{L} : \set{X} \times \alg{Z} \ni (x, \set{A}) \mapsto \int \kernel{K}(x, \rmd y) \, \kernel{L}(y, \set{A})
\end{align}
and
\begin{align}
	\kernel{K} \varotimes \kernel{L} : \set{X} \times (\alg{Y} \varotimes \alg{Z}) \ni (x, \set{A}) \mapsto \int \1{\set{A}}(y,z) \, \kernel{K}(x, \rmd y) \, \kernel{L}(y, \rmd z) 
\end{align}
whenever these are well-defined. These products form new transition kernels from $\measSpace{X}$ to $\measSpace{Z}$ and from $\measSpace{X}$ to $(\set{Y} \times \set{Z}, \alg{Y} \varotimes \alg{Z})$, respectively. Also the $\varotimes$-product
 of a kernel $\kernel{K}$ and a measure $\nu \in \meas{\alg{X}}$ is defined as the new measure
\begin{align}
	\nu \varotimes \kernel{K} : \alg{X} \varotimes \alg{Y} \ni \set{A} \mapsto \int \1{\set{A}}(x,y) \, \kernel{K}(x, \rmd y) \, \nu(\rmd x) \eqsp.
\end{align}
Finally, for any kernel $\kernel{K}$ and any bounded measurable function $h$ we write $\kernel{K}^2 h \eqdef (\kernel{K} h)^2$ and $\kernel{K} h^2 \eqdef \kernel{K}(h^2)$. Similar notation will be used for measures.

\end{document}